\documentclass[11pt]{amsart}
\usepackage{amsaddr}
\usepackage{graphicx}
\usepackage{color}

\usepackage[x11names, svgnames, rgb]{xcolor}

\usepackage{wrapfig}
\usepackage[all]{xy}

\usepackage{pgfplots}
\usepackage{tikz}
\usepackage{pgf}
\usetikzlibrary{circuits,calc,decorations,arrows,positioning, fit,arrows,automata,matrix,cd}

\usepackage{amsmath,amssymb,amsbsy,amsfonts,mathrsfs}
\usepackage{euscript,stmaryrd}
\usepackage{enumerate}
\usepackage{xcolor}
\usepackage{blindtext}
\usepackage[normalem]{ulem}

\definecolor{vertrep}{HTML}{30A930}
\definecolor{bleuajout}{HTML}{0000DE}
\definecolor{jaunecom}{HTML}{FFFFC0}

\newcommand{\N}{{\mathbb{N}}}

\newcommand{\A}{{\tt A}}
\newcommand{\B}{{\tt B}}

\newcommand{\proj}{pr} 

\newcommand{\arete}[3]{{#1 \xrightarrow{#2} #3}}
\newcommand{\path}[5]{#1 \overset{#2}{\longrightarrow} #3 \longrightarrow \cdots \overset{#4}{\longrightarrow} #5}

\newcommand{\triple}[3]{\underset{#1}{\overset{#3}{#2}}}

\newcommand{\Amin}{\A_{\rm{min}}}

\newcommand{\DG}{\boldsymbol{\mathrm{Dir}}}

\newcommand{\DSG}{\boldsymbol{{\mathrm{sDir}}}}
\newcommand{\G}{\boldsymbol{\mathrm{Gr}}}

\newcommand{\Emu}{\boldsymbol{\mathrm{Em}}}
\newcommand{\EmuS}{\boldsymbol{\mathrm{Em}_S}}
\newcommand{\EmuO}{\boldsymbol{\mathrm{Em}^{0}}}
\newcommand{\Cov}{\boldsymbol{\mathrm{Cov}}}
\newcommand{\CovS}{\boldsymbol{\mathrm{Cov}_S}}
\newcommand{\Exc}{\mathrm{Exc}}

\newcommand{\Maut}{\mathbf{Auto}}

\newcommand{\op}{\mathrm{op}}

\newcommand{\double}[1]{ {\overleftrightarrow{#1}}}

\newcommand{\Gaut}[1]{{{\rm A}_{#1}}}
\newcommand{\Autg}[1]{ {\rm G}_{#1}}

\newcommand{\relab}[2]{{#1}_{\backslash {#2}}}
\newcommand{\outE}{{\rm outE}}

\newtheorem{theorem}{Theorem}
\newtheorem{prop}{Proposition}
\newtheorem{corollary}{Corollary}[theorem]
\newtheorem{lemma}{Lemma}
\newtheorem{corolem}{Corollary}[lemma]

\newtheorem{conj}{Conjecture}

\theoremstyle{definition}
\newtheorem{definition}{Definition}

\theoremstyle{remark}
\newtheorem{example}{Example}
\newtheorem{remark}{Remark}

\tikzstyle{etat}=[circle, draw=black]
\tikzstyle{etat2}=[circle, inner sep= 2pt, minimum size=5pt, draw=black]

\begin{document}

\title[The genus of regular languages and directed graph emulators]{The genus of regular languages and directed graph emulators}
\author{Guillaume Bonfante${}^1$, Florian Deloup${}^2$	 
}
\address{1- Universit\'e de Lorraine - LORIA - Nancy\\
	2- Universit\'e Paul Sabatier - IMT - Toulouse}
\date{}

\begin{abstract}
The article continues our study of the genus of a regular language $L$, defined as the minimal genus among all genera of all finite deterministic automata recognizing $L$. Here we define and study two closely related tools on a directed graph: directed emulators and automatic relations.
A directed emulator morphism essentially encapsulates at the graph-theoretic level an epimorphism onto the minimal deterministic automaton. An automatic relation is the graph-theoretic version of the Myhill-Nerode relation. We show that an automatic relation determines a directed emulator morphism and respectively, a directed emulator morphism determines an automatic relation up to isomorphism. Consider the set $S$ of all directed emulators of the underlying directed graph of the minimal deterministic automaton for $L$. We prove that the genus of $L$ is $\underset{G \in S}{\min}\ g(G)$. We also consider the more restrictive notion of directed cover and prove that the genus of $L$ is reached in the class of directed covers of the underlying directed graph of the 
 minimal deterministic automaton for $L$. This stands in sharp contrast to undirected emulators and undirected covers which we also consider. Finally we prove that if the problem of determining the minimal genus of a directed emulator of a directed graph has a solution then the problem of determining the minimal genus of an undirected emulator of an undirected graph has a solution.
\end{abstract}

 \maketitle

 \tableofcontents

\section{Introduction}

Regular languages form the simplest class of languages in Chomsky's hierarchy: these are the languages recognized by deterministic finite automata (DFA). A special r\^ole is played by the minimal automaton canonically associated to a regular language. The complexity measure provided by the latter is familiar: it allows, for instance, to compare two regular languages by their size or to decide whether one is included into the other. Yet there are questions that cannot be read directly from the minimal automaton, and there exist complexity measures that are not captured simply by the size of the minimal automaton. An instance of the latter is provided by the genus of a regular language. The genus of a regular language $L$ is the minimal genus of a finite deterministic automaton recognizing the language $L$. In particular, a planar language is a language with genus $0$. We defined the genus in~\cite{BD16} and showed that this notion defines a proper hierarchy of (regular) languages. The decidability of the planarity of a regular language is a question asked by R.~V.~Book and A.~K.~Chandra \cite{BC76} in 1976.

The genus of a regular language behaves quite differently from the its set-theoretical size. First, a minimal genus deterministic automaton recognizing a given language has not minimal size in general, thus differs from the minimal automaton \cite{BD16}. Secondly, in a sequel \cite{BD19} to the first article, we proved that the size of the minimal automaton recognizing a language $L$ with minimal genus may be exponentially larger than the minimal automaton itself. The existence of an upper bound remains open. This suggests that the computation of the genus of a regular language should have a high complexity.

We also proved that under a fairly generic hypothesis\footnote{The minimal automaton must not contain small cycles, the size of which depends on the size of the alphabet. Moreover, automata are supposed to be complete.}, the genus of a regular language is computable. In particular, under this hypothesis, the planarity of a regular language is decidable. One should point out that the proof is entirely constructive and yields an implementable algorithm. We conjecture that the genus of a regular language is computable in general.



In this third opus of the series, we look at the problem with a graph-theoretical approach. The graph-theoretical substance of the relation of the minimal automaton to the genus minimal automaton is the notion of \emph{minimal directed emulator} (\S \ref{sec:emu}). A particular case is the notion of {\emph{directed cover}}. Most of this paper is devoted to the systematic study of directed emulators.

The notion of directed emulator is the natural refinement of the notion of (undirected) graph emulator, introduced by R.~M.~Fellows in 1985. In order to study the properties of emulators, he also used the notion of graph cover and conjectured that a connected finite graph has a finite planar emulator if and only if it has a finite planar cover. It took more than twenty years before Y.~Rieck and Y.~Yamashita found a counterexample \cite{RY09} to the conjecture. However, P.~Hlin\v{e}n\`y had found ten years earlier \cite{Hl99} an example of a graph having an emulator of genus strictly less than the minimal genus of any of its covers. On the one hand, we show that the situation for directed graphs is much simpler: a directed graph has a directed emulator of genus $g$ if and only if it has a directed cover of genus $g$ (Proposition~\ref{prop:directed_is_cover}). On the other hand, it follows from our constructions (Corollary~\ref{cor:sol_for_directed_implies_sol_for_undirected}) that a solution to the directed emulation genus problem implies a solution to the undirected emulation genus problem, a well-known difficult problem even in genus $0$. 

A major result of this paper is that a regular language has genus less than or equal to $g$ if and only if a directed cover of the underlying graph of its minimal deterministic automaton has genus less than or equal to $g$ (Theorem~\ref{th:genus-and-emulator}). As a consequence, we show that the original problem of determining the genus of a regular language is equivalent to the problem of determining the minimal genus of a directed cover of the underlying directed graph (as opposed to the undirected version). In particular, the problem of determining the planarity of a regular language is equivalent to the problem of determining the planarity of a directed cover of the underlying directed graph.  
The latter equivalence was originally proved by D.~Kuperberg in~\cite{Denis} via a different approach. 

Another important result of the paper is the correspondence between directed emulators and automatic relations on directed graphs as stated in Theorem~\ref{th:unique_automatic_decomposition}. Roughly speaking, an automatic relation on a directed graph is a device that mimicks the Myhill-Nerode equivalence relation at the level of automata. There is no privileged subset of edges (sets of initial and final states respectively) on a directed graph $G$ but only a class of complete final systems. It is proved here (Theorem~\ref{th:automatic_as_mn}) that a relation in a directed graph is automatic if and only if it stems from a MN-recursive relation, reminiscent of the recursive description of the original Myhill-Nerode relation. In the particular (important) case when the directed graph has a reachable vertex, we show that an automatic relation is induced by a Myhill-Nerode relation with respect to a single distinguished subset of vertices (Corollary~\ref{cor:accessible_automatic_description}).
 
Four main mathematical structures arise. First, (finite state) automata are used to describe regular languages. Forgetting input states and final states (albeit not completely, cf. the r\^ole of complete final systems in Theorem~\ref{th:automatic_as_mn} and Corollary~\ref{cor:accessible_automatic_description}) yields the notion of semi-automata. Forgetting letters on edges leads to the notion of directed graphs. Finally, forgetting the direction on edges leads to undirected graphs. These three forgetful operations are performed when the genus is computed. Our main contribution is to delineate the relationships between these four structures under the lens of the genus problem. 

As a whole, our approach aims at extracting the minimal structure required to compute the genus of a regular language. We use categories which we think render correctly the structural framework of the problem. Many authors have developped a categorical approach to closely related structures. For instance, J.~Ad\'amek and V.~Trnkov\'a~\cite{Adamek} present a categorical view of automata. T.~Colcombet and D.~Petri\c{s}an in \cite{ColcombetP19} develop a nice theory in which automata are seen as functors. They provide sufficient conditions to ensure the existence of minimal automata. We should also mention the work by A.~Joyal, M.~Nielsen and G.~Winskell about concurrency modelling in \cite{JoyalNW96}, where their transition systems suppose input states (actually a unique one).
This is an example of structure one naturally meets in our setting. Several obstacles lie in our way to build a full fledged categorical theory. First, there are two sorts of morphisms, emulators and covers. 
Secondly, to each directed emulator, one may associate a directed cover with the same genus (one should keep in mind that this is not true for the undirected graphs) but this construction is not natural. As another example, removing parallel edges can be done with emulators, not with covers. On the other hand, a directed cover lends itself to a description in terms of a universal property, see for instance Proposition~\ref{prop:directed_is_cover}. The subtle interplay between the two lies at the heart of the complexity analysis of the genus of a regular language.

\emph{Acknowlegments}. F.D. wishes to thank Philippe Balbiani for an illuminating discussion about bisimulation at the Institut de recherche en informatique de Toulouse (IRIT) in 2014. G.B. thanks Fran\c{c}ois Lamarche for some technical hints. Both authors would like to thank Denis Kuperberg for preliminary work in 2015 and for pointing out an error in an earlier version of this manuscript.

\renewcommand{\thetheorem}{\Alph{theorem}}
\section{Outline}

Given an automaton $\A$, let $L(\A)$ denote the language recognized by $\A$ and let ${\rm{G}}_{\A}$ denote the underlying simple directed graph. The \emph{genus} $g(L)$ of a regular language $L$ is the minimal genus of a deterministic automaton recognizing the language:
$$ g(L) = \min \ \{ g({\rm{G}}_{\A}) \ | \ L(\A) = L, \A\ {\rm{deterministic}} \}.$$

A directed emulator morphism between directed graphs is a graph morphism $\pi$ that sends the outgoing edges at each vertex $w$ surjectively onto the outgoing edges at $\pi(w)$. More precisely, a directed graph $G'$ is a \emph{directed emulator} of a directed graph $G$ if there is a directed graph epimorphism $G' \to G$ such that for any vertex $v$ of $G$, any outgoing edge $e$ starting at $v$ and any vertex $v'$ of $G'$ in the preimage of $v$, there is an outgoing edge $e'$ starting at $v'$. A more restrictive notion is that of \emph{directed cover} where the outgoing edges at $w$ are sent bijectively onto the outgoing edges at the image of $w$.

Let $L$ be a regular language, let $\A_{\rm{min}}(L)$ be the Myhill-Nerode minimal automaton recognizing $L$ and let $G(L) = {\rm{G}}_{\A_{\rm{min}}}(L)$ be the underlying directed graph.

In a first step, we give alternative descriptions of emulators and covers in terms of automatic relations on a directed graph. These will be Theorem~\ref{th:unique_automatic_decomposition} and \ref{th:automatic_as_mn} for emulators and Proposition~\ref{pr:presentation_of_covers} for covers. 

From there, we describe the lattice structure underlying emulators which will lead to minimal objects (in the sense of their size) in the spirit of~\cite{ColcombetP19}.

Putting these together, we prove Theorem~\ref{th:genus-and-emulator} that states that the three following statements are equivalent:
\begin{enumerate}[(i)]
\item the regular language $L$ has genus $g$
\item the directed graph $G(L)$ has a directed emulator of genus $g$
\item the directed graph $G(L)$ has a directed cover of genus $g$.
\end{enumerate}
This statement can be reformulated as an equivalence of the decidability of the genus problem for regular languages and directed graphs. 

Throughout the paper, some statements are asserted without proof. We do it only when proofs are immediate consequences of the definitions. Details are provided whenever a more elaborate proof is required or a specific construction. 

\renewcommand{\thetheorem}{\arabic{theorem}}
\setcounter{theorem}{0}

\section{Preliminary material} \label{sec:prelim}

\label{truc}

\subsection{Graphs}

Throughout this paper, $V$ and $E$ are finite sets. Given a set (resp. a category) $A$, we denote by $1_{A}$ the identity (resp. the identity functor) on $A$. Given a set $A$, we usually denote by $\sim$ an equivalence relation on $A$. If $\sim'$ is another equivalence relation, we write $\sim \subseteq \sim'$ if $x \sim y$ implies $x \sim' y$.

A \emph{directed graph} (sometimes shortened to \emph{digraph}) $G$ consists of a set $V$ of {\emph{vertices}} and a set $E$ of {\emph{edges}} and two maps $s,t: E {\rightrightarrows} V$ (resp. ``source'' and ``target'').  

Given an edge $e$, the notation $\arete{x}{e}{y}$ states that $s(e) = x$ and $t(e)=y$. An edge $\arete{x}{e}{x}$ is a {\emph{loop}} at vertex $x$. Given a graph $G$, we denote by $V_G$ its vertices, by $E_G$ its edges and by $s_G, t_G$ its corresponding source and target functions. We shall drop subscripts if the context is clear.

A \emph{walk} of length $k \geq 1$ in a graph $G$ starting at $v \in V$ and ending at $w \in V$ is a sequence $e_1,\ldots, e_k$ of edges with $s(e_{i+1})= t(e_i)$ for $i\leq k-1$ such that $v=s(e_1)$ and $w = t(e_k)$. It will be convenient to regard a single vertex $v \in V$ as a walk of length $0$. A vertex $w \in V$ is \emph{reachable from a vertex} $v \in V$ if there is a walk $e_1, \ldots, e_k$ in $G$ such that $s_{G}(e_1) = v$ and $t_{G}(e_k) = w$. In particular the target vertex is reachable from the source vertex. The relation $v \leq w$ if and only if $w$ is reachable from $v$ is a preorder on the set $V$ of vertices. A vertex $v$ is {\emph{reachable}} if $v$ is reachable from any vertex in $V$. A directed graph $G$ is {\emph{reachable}} if there is at least one reachable vertex in $G$. A vertex $v$ is {\emph{co-reachable}} if any vertex in $V$ is reachable from $v$. A directed graph $G$ is {\emph{co-reachable}} if there is at least one co-reachable vertex in $G$.

Given a directed graph $G$, the \emph{ordered boundary map} is the map  $\Delta_G: E_G \to V_G \times V_G$ defined by $\Delta_{G}(e)=(s_G(e),t_G(e))$. An edge $\arete{x}{e}{y} \in E_G$ is {\emph{simple}} if there is no other edge $e'$ such that $\Delta_G(e) = \Delta_G(e')$. The graph $G$ is \emph{simple} if any of its edges is simple.

A {\emph{morphism}} $G \to H$ {\emph{between directed graphs}} $G$ and $H$ is a pair $(p, q)$ where $p:V_{G} \to V_{H}$ and $q:E_{G} \to E_{H}$ satisfy the relations:
\begin{equation}
p \circ s_{G} = s_{H} \circ q\ \ {\rm{and}}\  \ p \circ t_{G} = t_{H} \circ q.
\label{eq:incidence_relation}
\end{equation}

Setting $p^{\times 2}:V_{G} \times V_{G} \to V_{H} \times V_{H}, \ (x,y) \mapsto (p(x), p(y))$, the previous equation is equivalent to: $p^{\times 2} \circ \Delta_G = \Delta_H \circ q.$ 
This relation is referred as the {\emph{adjacency relation}}.

The identity $(1_{V_G}, 1_{E_G}) : G \to G$ is a morphism and morphisms compose. 

\begin{lemma}
Directed graphs and morphisms between
them form a category denoted $\DG$. Directed simple
graphs and morphisms between them form a full subcategory $\DSG$ of $\DG$.
\end{lemma}

As defined, the category of graphs is equivalent to the homset category ${\rm Hom}_{\mathbf{Fset}}(\begin{tikzpicture}[->,transform shape, scale=0.7, baseline=-1mm]\node (A0) at (0,0) {.};
\node (A1) at (1,0) {.};
\draw[transform canvas={yshift=0.5mm}] (A0) edge node[above]{s} (A1);
\draw[transform canvas={yshift=-0.5mm}] (A0) edge node[below]{t} (A1);
\end{tikzpicture},-)$ where ${\mathbf{Fset}}$ is the category of finite sets and maps between them. We do not use it explicitly but the equivalence occurs between the lines in the sequel. 

A {\emph{graph epimorphism}} (resp. {\emph{monomorphism}}, {\emph{isomorphism}}) is a graph morphism $(p,q)$ such that both maps $p$ and $q$ are surjective (resp. injective, bijective).

A subset $W \subseteq V$ of vertices of a graph $G = (V, E, s, t)$ determines the graph
$G_{\mid W} = (W, E_{W}, s|_{E_{W}}, t|_{E_{W}})\ {\hbox{where}}\ E_{W} = \{ e \in E \ | \Delta(e)  \in
W\times  W \}$. 
Similarly, a subset $F \subseteq E$ of edges determines the graph
$G_{\mid F} = (V_G, F, s|_{F}, t|_{F})$. A directed {\emph{subgraph}} of $G$ is a graph $H$ such that there is a set of vertices $W$ and a set of edges $F$ such that $H = (G_{\mid W})_{\mid F}$. We also say that $G$ contains $H$. If $\phi:G \to H$ is a monomorphism, then $G$ is isomorphic to some directed subgraph of $H$.

Given a morphism $(p,q) : G \to H$, we define its image graph to be $(p,q)(G) = (p(V_G), q(E_G), s, t)$ with $s(q(e)) = p(s_G(e))$ and $t(q(e)) = p(t_G(e))$ for all $e \in E_G$. It is clear that $(p,q)(G)$ is a subgraph of $H$ and that $(p,q) : G \to (p,q)(G)$ is an isomorphism. 

Given a directed graph $G$, consider the directed graph defined $R(G) = (V_G, \Delta_G(E_G), \pi_1, \pi_2)$ with $\pi_1(x,y) = x$ and $\pi_2(x,y) = y$.

\begin{lemma}\label{lem:Delta_is_identity} $\Delta_{R(G)} = 1_{E_{R(G)}}$.
\end{lemma}


\begin{corolem}
The graph $R(G)$ is simple.
\end{corolem}

\begin{corolem}
The assignment $G \mapsto R(G)$ extends to a functor $\DG \to \DSG$ that assigns to a morphism $(p,q):G \to G'$ the morphism $R(p,q):R(G) \to R(G')$ defined by $R(p,q) = (p, p^{\times 2})$.
\end{corolem}

\begin{proof}
An edge $\Delta_{G}(e) \in E_{R(G)}$ is mapped to $p^{\times 2} \circ \Delta_{G}(e)$. This is a morphism since the adjacency relation is satisfied, for $\Delta_{R(G')} \circ p^{\times 2} = p^{\times 2} = p^{\times 2} \circ \Delta_{R(G)}$ according to Lemma~\ref{lem:Delta_is_identity}.
\end{proof}


%
%

\begin{lemma} \label{lem:R_is_epi} The pair $\rho_G = (1_{V_G}, \Delta_G) : G \to R(G)$ is an epimorphism. It is an isomorphism when restricted to simple graphs.
	Moreover, $G \mapsto \rho_G$ defines a natural transformation $1_{\DG} \to I \circ R$ where $I$ denotes the inclusion functor $\DSG \to \DG$:
	$$\xymatrix{
		G' \ar[r]^-{\rho_{G'}} \ar@{->}_{\psi}[d]  & R(G') \ar^{R(\psi)}[d] \\
		G \ar[r]^-{\rho_G}  & R(G)
	}$$
	\end{lemma}

\begin{lemma} \label{lem:R_is_invol} $R \circ I \circ R = R$.
\end{lemma}

There is yet another endofunctor we shall need. Let $G = (V,E,s,t)$ be a directed graph. Then $G^{\op} = (V,E,t,s)$ is another directed graph. 
 A morphism $(f,g): G \to G'$  is sent to  $(f,g)^{\op} = (f,g) :G^{\op} \to G'^{\op}$.

\begin{example} \label{example:op}
Let $G = \! \! \! \begin{tikzpicture}[->,transform shape, scale=0.6, baseline=-1mm]
 \tikzset{vertex/.style={circle, minimum size=12pt,inner sep=8pt}}
\node[etat] (G0) at (0,0) {\Large $v$};
\node[etat] (G1) at (2,0) {\Large $w$};
\path[] (G0) edge [loop left, looseness=15] node[above] {$g$} (G0);
\path[] (G0) edge [bend left=15] node[above] {$e$} (G1);
\path[] (G1) edge [bend left=15] node[below] {$f$} (G0);
\end{tikzpicture}$, then $G^{\rm{op}}=\!\!\begin{tikzpicture}[->,transform shape, scale=0.6, baseline=-1mm]
 \tikzset{vertex/.style={circle, minimum size=12pt,inner sep=8pt}}
\node[etat] (G0) at (0,0) {\Large $v$};
\node[etat] (G1) at (2,0) {\Large $w$};
\path[] (G0) edge [loop left, looseness=15] node[above] {$g$} (G0);
\path[] (G0) edge [bend right=15] node[below] {$f$} (G1);
\path[] (G1) edge [bend right=15] node[above] {$e$} (G0);
\end{tikzpicture}.$ In the drawings, $\{v, w\}$ are the vertices, $\{e, f, g\}$ the edges and arrows describe sources and targets.
\end{example}

We record a few properties of the graphs that are ``invariant'' under ${\rm{op}}$.

\begin{lemma} \label{lem:properties-op}
The following relations hold:
\begin{align}
(G^{\op})^{\op} = G. \label{eq:op-involutive} \\
R(G^{\op}) = R(G)^{\op}. \label{eq:G-and-op-commute}
\end{align}
\end{lemma}


We record the following fact which we shall refine later.

\begin{lemma} \label{lem:functors_R_and_op}
The functor $R:\DG \to \DSG$ is essentially surjective and full.
The endofunctor ${(-)}^{\rm{op}}$ is both full and faithful.
\end{lemma}

\begin{proof}
$R$ is essentially surjective since $\rho_{G}: G \to R(G)$ is an isomorphism for simple graphs. Fullness is a direct consequence of Lemma~\ref{lem:R_is_invol}.  The second statement follows from
the identity (\ref{eq:op-involutive}).
\end{proof}

Excising loops in a directed graph consists in removing all loops from the set of edges
while keeping other edges and the set of vertices. 

\begin{definition}
Let $G$ be a directed graph. Let $L = \{ e \in E_G \ | \ s_G(e) = t_G(e) \}$ be the subset of loops. The
{\emph{excision}} of $G$ is the graph $G$ with all loops removed:
$${\rm{Exc}}(G) = G_{\mid E_G - L}.$$
\end{definition}

\begin{remark} \label{rem:exc_is_not_functorial}
The excision is not functorial since it cannot be extended to {\emph{graph morphisms}}. For instance, there is a graph morphism $\begin{tikzpicture}[->,transform shape, scale=0.6, baseline=-1mm]
\node[state, inner sep=4pt,minimum size=12pt](G0) at (0,0) {};
\node[state, inner sep=4pt,minimum size=12pt](G1) at (1.2,0) {};
\path [] (G0) edge node [] {} (G1);
\node[state, inner sep=4pt,minimum size=12pt](G0) at (3.0,0) {};
\path [] (G0) edge[loop right, looseness=15] node [] {} (G0);
\path[->,dotted] (1.7,0) edge (2.5,0);
\end{tikzpicture}$ but not if we remove the loop. We shall see later that in the appropriate category, Exc becomes a functor.
\end{remark}

\subsection{Undirected graphs}

An \emph{undirected graph} $G = (V,E,\partial)$ consists of a set $V$ of {\emph{vertices}}, a set $E$ of {\emph{edges}} and a map $\partial: E \to {\mathcal{P}}_{2}(V)$ from the edges to the set of unordered pairs of vertices\footnote{ We define ${\mathcal{P}}_2(V) = \{ \{v, w\} \mid v, w \in V\}$. So, formally, $X \in {\mathcal{P}}_2(V)$ is either a singleton or a pair.}. 
A {\emph{morphism}} $G \to H$ \emph{between undirected graphs} is a pair of maps $p: V_{G} \to V_{H}$ (between vertices) and $q:E_{G} \to E_{H}$ (between edges) such that:
\begin{equation}
\partial_{H} \circ q = p^{\otimes 2} \circ \partial_{G}
\label{eq:incidence_for_undirected}
\end{equation}
with $p^{\otimes 2}( \{x,y\} ) = \{p(x), p(y)\}$.

Undirected graphs and morphisms between them form a category denoted by $\G$. The canonical projection $\proj_V: V \times V \to {\mathcal{P}}_{2}(V)$ mapping $(v,w) \mapsto \{v, w\}$ induces a forgetful functor $U : \DG \to \G$. It maps objects $G \mapsto (V_G, E_G, \proj_{V_G} \circ \Delta_G)$ and it acts as the identity on morphisms. 

\begin{lemma} \label{lem:forget_op}
$U(G^{\rm{op}}) = U(G)$.
\end{lemma}

\subsection{Semi-automata (labelled graphs)}

We consider here an enrichment of graphs where one associates a label to each edge. We call these semi-automata. They are sometimes called labelled graphs or transition systems (see for instance~\cite{JoyalNW96}). For us, the structure is in-between graphs and automata, thus their denomination. For basic definitions about automata, we refer to \cite{Eilenberg} and \cite{Sakarovitch}. 

\begin{definition}
A  {\emph{semi-automaton}} $\A$ is a 3-tuple $\A = (G, \Sigma, \ell)$ made of
 \begin{itemize}
 \item a directed graph $G$ whose vertices and edges are respectively called \emph{states} and \emph{transitions},
 \item a fixed set $\Sigma$, the \emph{alphabet}, together with a surjective map $\ell: E \to \Sigma$ (labelling of the edges in $A$).
\end{itemize}
\end{definition}

 Again, given a semi-automaton $\A$, we write $G_\A$ its underlying graph, an operation that will turn out to be functorial. We use freely $V_\A$ for states, $E_\A$ for transitions, etc.

\begin{remark}
Some authors do not assume the labelling $\ell$ to be onto. In this case, we say the alphabet is {\emph{underused}}. In this paper, we assume that all semi-automata have no underused alphabet. The hypothesis is required for Lemma~\ref{lem:forgetful_faithful_functor}. 
\end{remark}

\begin{example}In the drawing of a semi-automaton $\A$,
\begin{tikzpicture}[->,transform shape, scale=0.7, initial text={}, baseline=-1mm]
\node[state, inner sep=3pt, minimum size=12pt ](G0) at (0,0) {$q$};
\node[state,  inner sep=2pt, minimum size=12pt](G1) at (1.5,0) {$q'$};
\path [] (G0) edge [bend left=15] node [above=0.5mm] {$e:a$} (G1) ;
\path [] (G1) edge [bend left=15] node [below=0.5mm] {$e':a$} (G0);
\path [] (G0) edge [loop left] node [left=0.5mm] {$e'':b$} (G0);
\end{tikzpicture}, 
$e$ stands for the edge while $a$ denotes its corresponding label. 
\end{example}  

Let $\A$ and $\B$ be two semi-automata. A {\emph{morphism}} $(f, g, \alpha) : \A \to \B$ is a directed graph morphism $(f,g): G_\A \to G_\B$ together with 
a map $\alpha:\Sigma_\A \to \Sigma_\B$ between alphabets such that
\begin{equation}
\alpha \circ \ell_\A = \ell_\B \circ g. \label{eq:rel_labels_morphisms}
\end{equation}
If $\alpha$ is the identity, we say that the morphism is {\emph{strict}}. If $f$ and $g$ are identities, the morphism  is called a {\emph{relabelling}}.

The identity $(1_{V_\A}, 1_{E_\A}, 1_{\Sigma_\A}) : \A\to \A$ is a (strict) morphism. And it is easily seen that the composition of two morphisms (resp. strict morphisms)
 is a morphism (resp. a strict morphism).  We denote by  \textbf{Semi} the category of semi-automata and their morphisms and by $\textbf{Semi}^0$ the category of semi-automata with their strict morphisms. The category ${\mathbf{Semi}}^{0}$ is a faithful subcategory of ${\mathbf{Semi}}$.
 
 We defined semi-automata in a slightly unorthodox way because of the tropism towards graphs in this paper. Indeed, the following example
 	
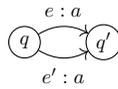
\begin{center}
 	\begin{tikzpicture}[->,transform shape, scale=0.7, initial text={}, baseline=-1mm]
 	\node (I0) at (0,0.75) {};
 	\node (F0) at (1.5,0.75) {};
 	\node (F1) at (0,-0.75) {};
 	\node[state, inner sep=3pt, minimum size=12pt ](G0) at (0,0) {$q$};
 	\node[state,  inner sep=2pt, minimum size=12pt](G1) at (1.5,0) {$q'$};
 	\path [] (G0) edge [bend left] node [above=0.5mm] {$e:a$} (G1) ;
 	\path [] (G0) edge [bend right] node [below=0.5mm] {$e':a$} (G1);
 	\end{tikzpicture}
 \end{center}
 is a semi-automaton for which there are two transitions labelled $a$ between $q$ and $q'$. This is not compatible with the standard definition of the transitions as a function $\delta : V \times \Sigma \to 2^{V}$ with $V$ the set of states and $\Sigma$ the alphabet. But actually, when we will restrict to deterministic automata, such a case will simply vanish.

 Let $\A$ be a semi-automaton, $W$ a subset of $V_\A$ and $E \subseteq E_\A$. Let $(\A_{\mid W})_{\mid E} = (G_{\mid W})_{\mid E}, \ell(\{e \in \tilde{E}\}), \ell_{e \in \tilde{E}})$ with $\tilde{E} = E_{(G_{\mid W})_{\mid E}}$. A semi-automaton produced in this fashion will be called a {\emph{sub-semi-automaton}} of $\A$. 

Given a semi-automaton morphism $(f, g, \alpha) : \A \to \B$, the semi-automaton $(f, g, \alpha)(\A) = ((f,g)(G_\A), \alpha(\Sigma_\A), {\ell_\B}_{|\alpha(\Sigma_\A)})$ is a sub-automaton of $\B$.

\begin{lemma}
Given a function $\alpha : \Sigma \to \Delta$ and a semi-automaton $\A = (G, \Sigma, \ell)$, the triple $\relab{\A}{\alpha} = (G, \alpha(\Sigma), \alpha \circ \ell)$ is a semi-automaton and furthermore, $(1_{V_\A}, 1_{E_\A}, \alpha) : \A \to \relab{\A}{\alpha}$ is a relabelling morphism. 
\end{lemma}

\begin{proof} Let $E = E_{\relab{\A}{\alpha}} = E_{\A}$.  We have $\alpha \circ \ell(E) = \alpha(A)$ since $\ell$ is surjective. So, $\relab{\A}{\alpha}$ is a semi-automaton. Second, we have $(\alpha \circ \ell) \circ 1_{E_{\relab{\A}{\alpha}}} = \alpha \circ \ell$, it is a morphism. 
\end{proof}

\begin{prop}\label{conformal_decomposition}
For any morphism $\phi$, there is a strict morphism $\pi$ and a relabelling morphism $\lambda$ such that $\phi = \lambda \circ \pi$ and a strict morphism $\psi$ and a relabelling $\mu$ such that  $\phi = \psi \circ \mu$. 
\end{prop}
\begin{proof} Given $(f, g, \alpha) : \A \to \B$, one of the decomposition is: $(f,g,\alpha) = (f,g, 1_\Sigma) \circ (1_V, 1_E, \alpha)$ where $V = V_{\A}$, $E = E_\A$ and $\Sigma = \Sigma_{\B}$. 

The other decomposition is: $(f,g,\alpha) = (1_{V_\B},1_{E_\B}, \alpha) \circ (f, g, \ell/ g)$ where $\ell/ g (g(e)) = \ell(e)$ for all $e \in E_\A$. The triple $(f, g, \ell/ g)$ is a proper morphism. Indeed, for all $g(e) \in E_{(f,g,\ell/ g)(\A)}$, we have $\alpha (\ell / g(g(e))) = \alpha(\ell_\A(e) ) = \ell_\B(g(e)) = \ell_\B(1_{E_\B}(g(e)))$. Thus, Equation~\ref{eq:rel_labels_morphisms} holds. 
\end{proof}

By definition, a semi-automaton $\A$ is a directed graph with extra structure. Forgetting this extra structure yields the {\emph{underlying}} directed graph $G_{\A}$ of the automaton.

\begin{lemma} \label{lem:forgetful_faithful_functor}
The assignment $\A \mapsto {\rm G}_{\A}$ yields a forgetful
faithful functor ${\rm G}_{(\mathunderscore)}:{\mathbf{Semi}} \to \DG$.
\end{lemma}

\begin{proof}
Let $\A $ and $\B$ be two semi-automata. Consider the induced map Hom$_{\mathbf{Semi}}(\A, \B) \to {\rm{Hom}}_{\DG}({\rm G}_{\A}, {\rm G}_{\B})$.
Suppose two morphisms $(f_i, g_i, \alpha_i):\A \rightrightarrows \B$, $i=0,1$, yield the same graph morphism $(f,g):{\rm G}_{\A} \to {\rm G}_{\B}$: they coincide at the level of the underlying graph hence the maps on the set $V$ of vertices in $\A$ and the set $E$
of edges in $\A$ coincide: $f_0 = f_1 = f$ and $g_0 = g_1 = g$.
Let $\alpha_i:\Sigma_\A \to \Sigma_\B$, $i=0,1$, be the respective maps between the alphabets. Both maps satisfy the relation $\alpha_i \circ \ell_\A = \ell_\B \circ g$, $i = 0,1$. Therefore $\alpha_0|_{\ell_\A(E)} = \alpha_1|_{\ell_\A(E)}$. Since $\A_0$ and $\A_1$
have no underused alphabet, $\alpha_0 = \alpha_1$.
\end{proof}

\begin{remark}
The functor ${\rm{G}}_{(\mathunderscore)}$ is not full. As an example, the induced map Hom$_{\mathbf{Semi}}$ $\left( \begin{tikzpicture}[->,transform shape, scale=0.7, initial text={}, baseline=-1mm]
\node[state, inner sep=3pt, minimum size=12pt ](G0) at (0,0) {$v$};
\node[state,  inner sep=2pt, minimum size=12pt](G1) at (1,0.4) {$u$};
\node[state, inner sep=2pt, minimum size=12pt](G2) at (1,-0.4) {$w$};
\path [] (G0) edge [bend left] node [above=0.5mm] {$e:a$} (G1);
\path [] (G0) edge [bend right] node [below=0.5mm]
{$e':a$} (G2);
\end{tikzpicture}{\mathbf{,}}\  \begin{tikzpicture}[->,transform shape, scale=0.7, initial text={}, baseline=-1mm]
\node[state, inner sep=3pt, minimum size=12pt ](G0) at (0,0) {$v$};
\node[state,  inner sep=2pt, minimum size=12pt](G1) at (1,0.4) {$u$};
\node[state, inner sep=2pt, minimum size=12pt](G2)at (1,-0.4) {$w$};
\path [] (G0) edge [bend left] node [above=0.5mm] {$e:a$} (G1);
\path [] (G0) edge [bend right] node [below=0.5mm]
{$e':b$} (G2);
\end{tikzpicture}
\right) \to {\rm{Hom}}_{\DG}\left( \begin{tikzpicture}[->,transform shape, scale=0.7, initial text={}, baseline=-1mm]
\node[state, inner sep=3pt, minimum size=12pt ](G0) at (0,0) {$v$};
\node[state,  inner sep=2pt, minimum size=12pt](G1) at (1,0.4) {$u$};
\node[state, inner sep=2pt, minimum size=12pt](G2) at (1,-0.4) {$w$};
\path [] (G0) edge [bend left] node [above=0.5mm] {$e$} (G1);
\path [] (G0) edge [bend right] node [below=0.5mm] {$e'$} (G2);
\end{tikzpicture}{\mathbf{,}}\  \begin{tikzpicture}[->,transform shape, scale=0.7, initial text={}, baseline=-1mm]
\node[state, inner sep=3pt, minimum size=12pt ](G0) at (0,0) {$v$};
\node[state,  inner sep=2pt, minimum size=12pt](G1) at (1,0.4) {$u$};
\node[state, inner sep=2pt, minimum size=12pt](G2) at (1,-0.4) {$w$};
\path [] (G0) edge [bend left] node [above=0.5mm] {$e$} (G1);
\path [] (G0) edge [bend right] node [below=0.5mm]
{$e'$} (G2);
\end{tikzpicture} \right)$ is not onto. For instance, the identity on the directed graph is not the image of a semi-automaton morphism.
\end{remark}

\begin{remark} \label{rem:surjectivity_is_preserved}
The functor ${\rm G}_{(\mathunderscore)}$ preserves surjectivity (resp. injectivity).
\end{remark}

\subsection{The genus of a graph}

We recall a few definitions. The {\emph{genus}} $g(\Sigma)$ {\emph{of a closed oriented surface}} $\Sigma$ is half the dimension of the first real homology vector space $H_{1}(\Sigma;{\mathbb{R}})$. Alternatively, it is the maximum number of mutually disjoint simple closed topologically nontrivial curves $C_1, \ldots, C_g$ such that the complement $\Sigma - (C_1 \cup \cdots \cup C_g)$ remains connected. This yields a natural
notion of genus of a graph.


\begin{definition}[Genus of a graph]
A graph has {\emph{genus}} $n$ if its geometrical realization is embeddable in a surface of genus $n$ but cannot be embedded in a surface of strictly smaller genus. We note $g(G)$ the genus of a graph $G$. 
\end{definition}

The definition makes sense for directed and undirected graphs alike.

\begin{lemma} \label{lem:genus_forget}
For any directed graph, $g(G) = g(U(G))$. 
\end{lemma}

\begin{proof}
The geometric realization only depends on the underlying undirected graph.
\end{proof}

 All the proofs dealing with the genus of some graph will rely on one of the following observations. 

\begin{lemma}\label{lem:genus_invariance_iso}
	If $G$ is isomorphic to $H$, then $g(G) = g(H)$. 
\end{lemma}

\begin{lemma} \label{lem:genus-invariance_opp}
The functor ${(-)}^{\rm{op}}$ preserves the genus.
\end{lemma}

\begin{proof}
For a directed graph $G$, $g(G^{\rm{op}}) = g(U(G^{\rm{op}}))$ (lemma~\ref{lem:genus_forget}) 
$ = g(U(G))$ (lemma~\ref{lem:forget_op}) $ = g(G)$ (lemma~\ref{lem:genus_forget}).
\end{proof}

\begin{lemma} \label{lem:removing-edge-does-not-increase-genus}
Removing an edge from a graph does not increase the genus: for any graph $G$ and $e \in E_G$,
$g(G - e) \leq g(G).$
\end{lemma}

\begin{lemma} \label{lem:R_preserves_genus}
The functor $R$ preserves the genus.
\end{lemma}

\begin{proof}
By induction from the previous lemma, $g(R(G)) \leq g(G)$. On the other hand, (the realization of) a double edge from a vertex $v$ to a vertex $w$ embeds as the boundary of a small bigon (a topological disc with two distinguished points, namely $v$ and $w$). Hence\footnote{By induction on the cardinality of the set $\Delta^{-1}(v,w)$ of multiple edges from $v$ to $w$ for each edge $(v,w) \in \Delta(E_{G})$.}, if $R(G)$ embeds in a surface $\Sigma$ then $G$ embeds there as well, so $g(R(G)) \leq g(G)$.
\end{proof}

\begin{lemma} \label{lem:excision-preserves-genus}
The excision operation preserves the genus: $g({\rm{Exc}}(G)) = g(G)$.
\end{lemma}

\begin{proof}
By Lemma~\ref{lem:removing-edge-does-not-increase-genus}, $g({\rm{Exc}}(G)) \leq g(G)$. On the other hand, (the geometric realization of) a loop at a given vertex embeds as the boundary of a small disc, hence if ${\rm{Exc}}(G)$ embeds in a surface $\Sigma$, then $G$ also embeds there, so $g(G) \leq g({\rm{Exc}}(G))$.
\end{proof}

\section{Directed emulators} \label{sec:emu}

In this section, we define directed emulators and study their properties. The basic material is introduced in \S \ref{subsec:basic_definitions}. Categorical and closure properties are presented in \S \ref{subsec:category_emulators}. Finally we briefly discuss the relation to undirected emulators in \S \ref{subsec:undirected_emulators}.

\subsection{Basic definitions} \label{subsec:basic_definitions}

The following definition is the main object of this section.

\begin{definition}\label{def:directed_emulator}
Let $\phi = (p,q): G \to H$ be a directed graph morphism. It is a \emph{directed emulator morphism} when:
	\begin{enumerate}[(i)]
		\item $p$ is surjective and
		\item  $\phi$ verifies the \emph{edge outgoing lifting property}.  That is, for any edge $e \in E_H$ and any vertex $x' \in V_G$ such that $p(x') = s_H(e)$, there is an edge $e' \in E_G$ such that $q(e') = e$ and $s_G(e') = x'$.
	\end{enumerate}

We say that the directed emulator morphism $\phi$ is a \emph{directed cover morphism} whenever in clause (ii), the edge $e'$ is unique.  
\end{definition}

If $\phi : G \to H$ is a directed emulator morphism (or sometimes shorter: a directed emulator), we say that $G$ is a directed emulator of $H$ and that $H$ is a directed amalgamation of $G$. Finally, we say that the edge $e'$ in the definition emulates the edge $e$.

\begin{example}\label{ex:iso_is_directed_emulator}
The identity $1_G : G \to G$ is a directed emulator morphism. More generally, an isomorphism $\psi : G \to H$ is a directed emulator morphism. It is also a directed cover morphism. 
\end{example}

\begin{remark}\label{re:cover_not_emulator}A directed emulator shall not be a directed cover as shown by the morphism:
	\begin{tikzpicture}[->,transform shape, scale=0.7, baseline=-2mm]]
	\node[etat](G0p) at (0.7,0) {$u$};
	\node[etat](H0) at (3.5,0) {$v$};
	\path (G0p) edge[loop left] node[right=1mm] {} (G0p);
	\path (G0p) edge[loop right] node[left=1mm] {} (G0p);
	\draw (1.9,0)[dotted] -- (2.8,0);
	\path (H0) edge[loop right] node[right=1mm] {} (H0);
	\end{tikzpicture}.
\end{remark}

\begin{example} \label{example:multi-to-simple-emulates}
For any directed graph $G$, the map $\rho_{G}:G \to R(G)$ is a directed emulator morphism.
\end{example}

\noindent The following observation is a direct consequence of the definition.

\begin{lemma} \label{lem:direct-emulator-is-epimorphism}
A directed emulator morphism is a directed graph epimorphism.
\end{lemma}

\begin{proof}
By definition, the vertex map is surjective. Given an edge $e$ in the base graph, since $p$ is surjective, there is a vertex $x \in p^{-1}(s_G(e) )$ and then, by the edge outgoing lifting property, there is an edge $e'$ such that $q(e') = e$.
\end{proof}

\begin{remark}\label{re:not_em_is_not_full}A directed graph epimorphism is not necessarily a directed emulator morphism. For instance, in the epimorphism below, 
 the central node $w_0$ has two outgoing edges while each element in its preimage ($u_0, v_0$) has only one.

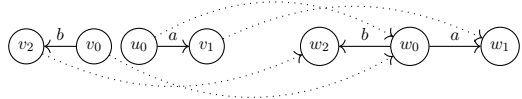
\centering{
 \begin{tikzpicture}[->,transform shape, scale=0.6, baseline=-2mm]]
\node[etat](G0) at (0,0) {$v_0$};
\node[etat](G0p) at (1,0) {$u_0$};
\node[etat] (G1) at (2.5,0) {$v_1$};
\node[etat](G2) at (-1.5,0) {$v_2$};
\node[etat](H0) at (7,0) {$w_0$};
\node[etat] (H1) at (9,-0) {$w_1$};
\node[etat](H2) at (5,-0) {$w_2$};
\path(G0p) edge node [above] {$a$} (G1)
 (G0) edge node [above] {$b$} (G2)
 (H0) edge node[above] {$a$} (H1)
(H0) edge node[above] {$b$} (H2) ;
 \draw (G0)edge[dotted,bend right=30] (H0);
\draw (G1)edge [dotted, bend left=25] (H1) ;
 \draw (G2)edge [dotted,bend right=25](H2);
 \draw (G0p) edge [dotted, bend left=30]  (H0);
 \end{tikzpicture} }
\end{remark}

\begin{example} \label{example:simple-example} Directed amalgamation can ``create" loops:

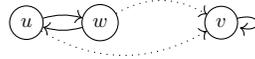
\centering{
\begin{tikzpicture}[->,transform shape, scale=0.7, baseline=-2mm]]
\node[etat](G0) at (3,0) {$v$};
\node[etat](G1) at (0.7,0) {$w$};
\node[etat] (G2) at (-0.7,0) {$u$};
\path(G1) edge[bend left=15] node [left=1mm] {} (G2)
 (G2) edge[bend left=15] node[left=1mm] {} (G1)
(G0) edge [loop right] node {} (G0);
\draw (G1) edge [dotted,bend left] (G0);
\draw (G2) edge [dotted,bend right]  (G0);
 \end{tikzpicture} }
\end{example}

\begin{remark} \label{rem:several-possible-emulator-maps}
There are in general several directed emulator morphisms between a digraph and a directed amalgamation of it. For instance, the digraph depicted below is a directed emulator of itself in two ways:
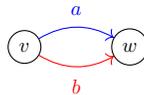
\begin{center}
\begin{tikzpicture}[->,transform shape, scale=0.7, baseline=-2mm]]
\node[etat](G0) at (0,0) {$v$};
\node[etat](G1) at (2,0) {$w$};;
\path(G0) edge[bend left, color=blue] node [above=1mm] {$a$} (G1)
 (G0) edge[bend right,color=red] node[below=1mm] {$b$} (G1);
 \end{tikzpicture}
\end{center}
The first one is the identity and the other one swaps the two edges. In both cases, the map on the vertices is the identity.
\end{remark}

\begin{definition} \label{def:centripetal_star}
Let $G$ be a directed graph and $x \in V_G$ a vertex. The {\emph{centripetal star}} of $x$ is the set
$\outE_{G}(x)$ {\emph{of
outgoing edges from}} $x$, that is
$\outE_G(x) = \{ e \in E_G \mid \arete{x}{e}{y} \text{ for some } y \in V_G\}$.
\end{definition}

A directed graph morphism $(p,q):G \to H$ induces, for each vertex $x \in V_G$, a map
$$ q_{x}: \outE_G(x) \to \outE_H(p(x)), \ e \mapsto q(e).$$

\begin{definition} \label{def:directed_immersion_submersion}
If at each $x \in V_G$, $q_{x}$ is injective (resp. surjective), then we say that $(p,q):G \to H$ is a directed immersion (resp. a directed submersion).
\end{definition}

A directed emulator morphism $(p,q): G' \to G$ lifts any outgoing edge $e \in E_{G}$ to an outgoing edge $e'$ from any specified vertex in the preimage $p^{-1}(s_G(e)) \in V_{G'}$. This implies the following observation.

\begin{lemma} \label{lem:emu-surjects-onto-outedges}
A directed graph epimorphism is a directed emulator morphism
if and only if it is a directed submersion. A directed cover is both a submersion and  an immersion. 
\end{lemma}

\begin{lemma}[Extraction of a directed cover] \label{lem:directed-emulator-contains-directed-cover}
	Let $G'$ be a directed emulator of a directed graph $G$. Then there is a directed subgraph $G''$ of $G$
	with the following properties:
\begin{enumerate}
\item[(1)] $G''$ is a directed cover of $G$;
\item[(2)] $V_{G'} = V_{G''}$.
\end{enumerate} 
\end{lemma}

\begin{proof}
	Let $p(x') = x$.  By Lemma~\ref{lem:emu-surjects-onto-outedges}, the emulator morphism induces a surjection
	${\rm{OutE}}(x') \to {\rm{OutE}}(x)$. Remove if necessary some outgoing edges from $x'$ so that a bijection ${\rm{OutE}}(x') \to {\rm{OutE}}(x)$ is obtained. Proceed thus on each
	vertex in the fibre $p^{-1}(x)$ for each vertex $x \in G$. The restriction of the directed emulator morphism is then a directed cover morphism.
\end{proof}
The construction is not universal. Here are two non isomorphic subgraphs verifying the conditions above:
$$
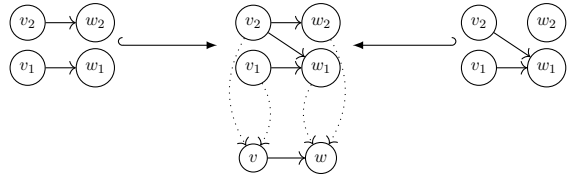
\begin{tikzpicture}[transform shape, scale=0.6,baseline=-1mm]
\node[etat] (A0) at (0,0) {$v$};
\node[etat] (B0) at (1.5,0) {$w$};
\node[etat] (A1) at (0,2) {$v_1$};
\node[etat] (B1) at (1.5,2) {$w_1$};
\node[etat] (A2) at (0,3) {$v_2$};
\node[etat] (B2) at (1.5,3) {$w_2$};
\node[etat] (A11) at (-5,2) {$v_1$};
\node[etat] (B11) at (-3.5,2) {$w_1$};
\node[etat] (A21) at (-5,3) {$v_2$};
\node[etat] (B21) at (-3.5,3) {$w_2$};
\node[etat] (A12) at (5,2) {$v_1$};
\node[etat] (B12) at (6.5,2) {$w_1$};
\node[etat] (A22) at (5,3) {$v_2$};
\node[etat] (B22) at (6.5,3) {$w_2$};
\draw[->] (A0) -- (B0);
\draw[->] (A1) -- (B1);
\draw[->] (A2) -- (B2);
\draw[->] (A2) -- (B1);
\draw[->] (A11) -- (B11);
\draw[->] (A21) -- (B21);
\draw[->] (A12) -- (B12);
\draw[->] (A22) -- (B12);
\path[->] (A1) edge[dotted,bend left] (A0);
\path[->] (A2) edge[dotted,bend right] (A0);
\path[->] (B1) edge[dotted,bend right] (B0);
\path[->] (B2) edge[dotted,bend left] (B0);
\draw[right hook-latex] (-3,2.5) -- (-0.8,2.5);
\draw[left hook-latex] (4.5,2.5) -- (2.2,2.5);
\end{tikzpicture}
$$


\begin{prop}\label{prop:directed_is_cover}
	Let $g \geq 0$. A directed graph has a directed emulator of genus $g$ if and only if it has a directed cover of genus $g$.
\end{prop}

\begin{proof}
	Lemma \ref{lem:directed-emulator-contains-directed-cover} shows that if a genus $g$ directed graph $G'$ emulates $G$, then it contains a directed subgraph $G''$ that covers $G$. Therefore $g(G'') \leq g(G') = g$. Therefore the class of directed emulators of $G$ contains a directed cover of minimal genus.
\end{proof}

We set $\mathbf{1} = \begin{tikzpicture}[transform shape, scale=0.6,baseline=-1mm]
\node[etat] (As) at (0,0) {$v$};\end{tikzpicture}$, $\mathbf{2} = \begin{tikzpicture}[transform shape,scale=0.6,baseline=-1mm]
\node[etat] (As) at (0,0) {$v$};
\node[etat] (Bs) at (1.5,0) {$w$};
\draw[->] (As) -- (Bs);
\end{tikzpicture}$ and the morphism $\mathbf{s12} : \mathbf{1} \to \mathbf{2} = [v \mapsto v]$. Definition~\ref{def:directed_emulator} can be reformulated as follows:
\begin{prop}\label{pr:diag_emu}
	A epimorphism $\phi : G \to H$ is a directed emulator  if and only if for all square diagram as below there is a morphism $e'$ making the two triangles commute.  The morphism is a cover if and only if there is at most one morphism $e'$ as a solution.
	\begin{center}	\begin{tikzpicture}[->]
		\node (As) at (0,0) {$\mathbf{1}$};
		\node (Bs) at (2.5,0) {$\mathbf{2}$};
		\node (A) at (0,-1.5) {$G$};
		\node (B) at (2.5,-1.5) {$H$};
		\draw[->] (As) -- node[above]{$\mathbf{s12}$}(Bs);
		\draw[->] (A) -- node[above]{$\phi$}(B);
		\path[->] (0,-0.4) edge node[left]{$x'$}  (A);
		\path[->] (2.5,-0.2) edge node[right]{$e$}  (B);
		\path[->] (2.3,-0.2) edge node[above]{$e'$} (A);
		\end{tikzpicture}\end{center}
\end{prop}

\subsection{The category of emulators} \label{subsec:category_emulators}

\begin{prop} \label{lem:composition_of_emulators}
	The composition of two directed emulators (resp. directed covering) morphisms is a directed emulator (resp. directed covering) morphism.
\end{prop}
The identities being directed emulators (resp. covers), see Example~\ref{ex:iso_is_directed_emulator}, directed graphs and directed emulator morphisms between them form a subcategory $\Emu$
of $\DG$. Simple directed graphs and directed simple emulator morphisms between them
form a subcategory of $\DSG$. The category of covers $\Cov$ comes with its subcategory $\CovS$ of covers over simple graphs.

\begin{lemma} To sum up, we have the inclusions:
	
	\begin{tikzpicture}[->]
	\node (As) at (0,0) {$\CovS$};
	\node (Bs) at (2,0) {$\EmuS$};
	\node (Cs) at (4,0) {$\DSG$};
	\node (A) at (0,-1.5) {$\Cov$};
	\node (B) at (2,-1.5) {$\Emu$};
	\node (C) at (4,-1.5) {$\DG$};		
	\draw[->] (As) -- node[above]{(1)} (Bs);
	\draw[->] (Bs) -- node[above]{(3)}(Cs);
	\draw[->] (A) -- node[above]{(2)}(B);
	\draw[->] (B) -- node[above]{(4)}(C);
	\draw[->,line width=2pt] (As) -- (A);
	\draw[->,line width=2pt] (Bs) -- (B);
	\draw[->,line width=2pt] (Cs) -- (C);
	\end{tikzpicture}
\end{lemma}

Thick arrows correspond to full subcategory relationships whereas thin ones correspond to simple inclusions. For (1) and (2), this is justified by Remark~\ref{re:cover_not_emulator}. For (3) and (4), we use Remark~\ref{re:not_em_is_not_full}.  

The graph \begin{tikzpicture}[transform shape, scale=0.6, baseline=-1mm] \node[etat] (L) at (0,0) {}; 
\path[->] (L) edge[loop right,in=340,out=20,looseness=15] (L); \end{tikzpicture}is a terminal object for $\DG$, but neither
for $\Emu$ nor $\Cov$ since the graph morphism \begin{tikzpicture}[transform shape, scale=0.6, baseline=-1mm] 
\node[etat] (M) at (-2,0) {};
\node[etat] (L) at (0,0) {}; 
\draw[->,dotted] (-1.5,0) to[] (-0.5,0) ; 
\path[->] (L) edge[loop right,in=340,out=20,looseness=15] (L); \end{tikzpicture}is neither an emulator nor a cover. 

The graph product $G \times H = (V_G \times V_H, E_G \times E_H, s_G \times s_H, t_G \times t_H)$ does not restricts to $\Emu$ nor $\Cov$. Indeed, the projection morphism $\pi_2$ in \begin{tikzpicture}[transform shape, scale=0.6, baseline=-1mm] 
\node[etat] (M1) at (-1.6,0) {v};
\node[etat] (L1) at (0.3,0) {v};
\node (L0) at (0.9,0) {$\times$};
\node[etat] (L2) at (1.5,0) {w};
\draw[->] (L2) edge[loop right] (L2);
\node[etat] (R1) at (5,0) {w};
\node (E) at (2.6,0) {=};
\node[etat] at (3.3,0) {vw};
\draw[->] (R1) edge[loop right] (R1);
\path[->,dotted] (-0.2,0) edge node[above]{$\pi_1$} (-1.1,0);
\path[->,dotted] (3.9,0) edge node[above]{$\pi_2$} (4.6,0);
\end{tikzpicture} is not an emulator.

 Given two morphisms $(p,q):G \to H$ and $(p',q') : G\to H$, the graph inclusion $G_{{\mid v \in V_G : p(v) =p'(v)}{\mid \{e \in E_G : q(e) = q'(e)\}}} \to G$ is an equalizer for the two morphisms in $\DG$. But some pairs of morphisms have no equalizers in $\Emu$. Consider for instance the graph $G=$\begin{tikzpicture}[transform shape, scale=0.6, baseline=-1mm] 
\node[etat] (L1) at (0,0) {v};
\node[etat] (L2) at (1.6,0) {w};
\path[->] (L1) edge[bend left=15] node[above] {$a$}  (L2)
(L1) edge[bend right=15] node [below] {$b$} (L2);
\end{tikzpicture} 
with its two morphisms $1_G$ and 
$Swap=[a \mapsto b, b \mapsto a]$. 
The graph inclusion \begin{tikzpicture}[transform shape, scale=0.6, baseline=-1mm] 
\node[etat] (M1) at (-3.6,0) {v};
\node[etat] (M2) at (-2.,0) {w};
\node[etat] (L1) at (0,0) {v};
\node[etat] (L2) at (1.6,0) {w};
\path[->] (L1) edge[bend left=15] node[above] {$a$}  (L2)
(L1) edge[bend right=15] node [below] {$b$} (L2);
\draw[->, dotted] (-1.5,0) -- (-0.5,0);
\end{tikzpicture} is an equalizer in $\DG$ but not in $\Emu$ nor $\Cov$. 

Since directed emulators are epimorphisms, we have:
\begin{prop}\label{pr:emu_right_cancel}
	Any morphism in $\Emu$ is right-cancellative: $\phi_1 \circ \psi = \phi_2 \circ \psi \Rightarrow \phi_1 = \phi_2$. 
\end{prop}

\begin{lemma}\label{lem:pull_backs} Given a morphism $\phi: G \to K$ in $\DG$ and a directed emulator $\phi' : H \to K$, in the pull back diagram:
$$
\xymatrix{
	G \times_{K} H \ar[r]^-{\pi_{2}|} \ar[d]_-{\pi_{1}|} & H \ar[d]^-{\phi'} \\
	G \ar[r]^-{\phi} & K
}
$$
the morphism $\pi_1|$ is a directed emulator. The property holds for covers. 
\end{lemma}

\begin{proof} For a direct calculation, we recall that $L = G \times_{K} H$ defined by
	\begin{align*}
	V_{L} = \{ (u,v) \in V_{G} \times V_{H}\ | \  p(u) = p'(v) \}, \\  
	E_{L} = \{ (e,f) \in E_{G} \times E_{H} \mid q(e) = q'(f) \},\\
	\Delta_{L}(e,f) = ((s_{G}(e),s_{H}(f)), (t_{G}(e), t_{H}(f)))
	\end{align*}
	is a pullback. But, let us consider the diagram below. We suppose the left square commutes (the right one also commutes being a pullback). 
	$$\begin{tikzcd}
	\mathbf{1} \ar[r, "x'"] \ar[d,"\mathbf{s12}"']& G \times_{K} H \ar[r,"\pi_{2}|"] \ar[d,"\pi_{1}|"']  & H \ar[d,"\phi'"]
	\\ 
	{\mathbf{2}} \ar[r,"e"'] \arrow[rru,bend right=10,"e'" near start]
	\ar[ru,bend left=10,"e''"]& G \ar[r,"\phi"'] & K
	\end{tikzcd}$$
	Then, since $\phi'$ is a directed emulator, there is an edge $e'$ such that $\phi' \circ e' = \phi \circ e$.  But, then, due to the pullback, given $e$ and $e'$, there is an edge $e''$ as shown. Thus, $\pi_1|$ is a directed emulator (Proposition~\ref{pr:diag_emu}). Thanks to the pull-back, it preserves the unicity, thus $e''$ is unique if $\phi'$ is a cover.
\end{proof}

\begin{corolem}\label{sub-graph-preserve}
	Given a subgraph $G$ of $K$ and a directed emulator morphism $H \to K$, there is a directed morphism $\pi : H'\to G$ with $H'$ a subgraph of $H$. The property holds for covers.
\end{corolem}

\begin{proof}Since pullbacks preserve monomorphisms, the morphism ${\pi_2}_{|}$ above is an monomorphism. Thus, $G \times_{K} H$ is isomorphic to a subgraph $H'$ of $H$. 
\end{proof}



\begin{lemma}\label{le:R_preserve_pullback}We have $R(G \times_K H) \simeq R(G) \times_{R(K)} R(H)$ for all graphs $G, K, H$. 
\end{lemma}

\begin{proof} On vertices, since $R$ acts as the identity, the result is immediate. On edges, the result follows from the definitions.  
\end{proof}

\begin{prop} \label{lem:functor_R_preserves_directed_emulators}
	The forgetful functor $R$ preserves directed emulators. 
\end{prop}

\begin{proof}
	Let $(p,q): G' \to G$ be a directed emulator.  Let us verify that $R(p,q) = (p, p^{\times 2}) : R(G') \to R(G)$ is a directed emulator. Let $\Delta_G(e)$ be an edge in $R(G)$ and let $x' \in R(G')$ such that $p(x') = s_{R(G)}(\Delta_G(e)) = s_{G}(e)$. Since $(p,q)$ is an emulator, there is $e' \in G'$ such that $q(e') = e$ and $s_{G'}(e') = x'$.  By definition, $s_{R(G')}( \Delta_{G'}(e') ) = s_{G'}(e') = x'$ and $p^{\times 2}( \Delta_{G'}(e')) = \Delta_G (q(e')) = \Delta_G(e)$. 
	
\end{proof}

\begin{remark}
	The functor $R$ does not preserve directed covers. For instance,  the directed graph $G = \begin{tikzpicture}[->,transform shape,scale=0.6, baseline=-1mm]]
	\node[etat] (H1) at (0,0){};
	\node[etat] (H2) at (1.5,0.4) {};
	\node[etat] (H3) at (1.5,-0.4){};
	\path [] (H1) edge node []{} (H2)
	(H1) edge node []{} (H3);
	\end{tikzpicture}$ covers the directed graph $H = \begin{tikzpicture}[->,transform shape,scale=0.6, baseline=-1mm]]
	\node[etat] (H1) at (0,0){};
	\node[etat] (H2) at (1.5,0) {};
	\path [] (H1) edge[bend left] node []{} (H2)
	(H1) edge[bend right] node []{} (H2);
	\end{tikzpicture}$. However, $R(G) = G$ only emulates (and does not cover) 
	$R(H) =  \begin{tikzpicture}[->,transform shape,scale=0.6, baseline=-1mm]]
	\node[etat] (H1) at (0,0){};
	\node[etat] (H2) at (1.5,0) {};
	\path [] (H1) edge node []{} (H2);
	\end{tikzpicture}.$
\end{remark}

That can be summed up by the diagram:
\begin{center}
	\begin{tikzpicture}[->]
	\node (As) at (0,0) {$\CovS$};
	\node (Bs) at (2,0) {$\EmuS$};
	\node (Cs) at (4,0) {$\DSG$};
	\node (A) at (0,-1.5) {$\Cov$};
	\node (B) at (2,-1.5) {$\Emu$};
	\node (C) at (4,-1.5) {$\DG$};		
	\draw[->] (As) -- node[above]{} (Bs);
	\draw[->] (Bs) -- node[above]{}(Cs);
	\draw[->] (A) -- node[above]{}(B);
	\draw[->] (B) -- node[above]{}(C);
	\draw[->] (As) edge [bend right] (A);
	\draw[->] (Bs) edge [bend right] (B);
	\draw[->] (Cs) edge [bend right] (C);
	\draw[->] (B) edge [bend right]  node[right]{R} (Bs);
	\draw[->] (C) edge [bend right] node[right]{R} (Cs);
	\end{tikzpicture}
\end{center}
The removal of loops (Exc) does not determine a functor in the category of directed graphs
(Remark \ref{rem:exc_is_not_functorial}).
The next observation is that Exc becomes one in the category $\boldsymbol{{\mathrm{Em}}^{0}}$ of directed graphs with directed emulators maps as morphisms. 

\begin{definition}
	 A {\emph{directed emulator map}} from a directed graph $G$ to an other directed graph $H$ is a pair $(p,q)$ made of a surjective map $p : V_{G} \to V_H$ and a \emph{partial} function $q:E_{G} \to E_H$ that is compatible with the adjacency relation and such that the pair $(p,q)$ has the edge outgoing lifting property. 
\end{definition}
Restricted to simple graphs, this yields a full subcategory $\boldsymbol{{\rm{Em}}_S^{0}}$ of $\EmuO$. 
A directed emulator morphism is a directed emulator map whose edge function is total. Thus, the category $\Emu$ is a subcategory of $\EmuO$.

Given a morphism $(p,q) : G \to H$, we define ${\rm Exc}(p,q) = (p',q') : \Exc(G) \to \Exc(H)$ with $p' = p$ and $q' = q_{\mid \{e \in E_G: p(s_G(e)) \neq p(t_G(e))\}}$. To justify that the definition is correct, we check three facts. First, the image of $q'$ stays within edges in $\Exc(H)$. Indeed, if $p(s_G(e)) \neq p(t_G(e))$, then $s_H(q(e)) \neq t_H(q(e))$ so that it is not a loop, thus in $\Exc(H)$. Second, being a restriction of $q$ by definition, $q'$ respect the adjacency relation. Third, if $e$ is an edge in $\Exc(H)$ and $x' \in \Exc(G)$ verifies $p(x') = s_H(e)$, there is an edge $e'$ in $G$ such that $q(e') = e$ and $s_G(e') = x'$. Since $e'$ is not a loop (otherwise $q(e')$ would be itself a loop), it is an edge within $\Exc(G)$.  

\begin{prop} \label{prop:exc_as_functor} $\Exc$ is a functor $\Emu \to \EmuO$. The excision ${\rm{Exc}}$ is a endofunctor of $\EmuO$ that restricts to an endofunctor of $\boldsymbol{{\rm{Em}}_S^{0}}$. 
	Furthermore,
	\begin{enumerate}
		\item[$(1)$] ${\rm{Exc}} \circ {\rm{Exc}} = {\rm{Exc}}$;
		\item[$(2)$] ${\rm{Exc}} \circ R = R \circ {\rm{Exc}}$.
	\end{enumerate}
\end{prop}

\begin{proof}
	Exc acts on a graph $G$ by removal of all loops; for any directed graphs $G,G'$, Exc acts on the set ${\rm{Hom}}_{\EmuO}(G',G)$ of directed emulator maps as the identity. The results follow.
\end{proof}

\begin{lemma}[Extension of a directed cover over an excized graph] \label{lem:from-cover-over-simple-excized-to-cover-over-original}
	Suppose that $\psi:G' \to G$ is a directed cover morphism between directed graphs. Furthermore, assume that $G = {\rm{Exc}}(H)$. Then there exists a directed graph $H'$ and a directed cover morphism $\varphi:H' \to H$ such that
	\begin{enumerate}
		\item[$(1)$]  ${\rm{Exc}}(H') = G'$;
		\item[$(2)$] The directed cover morphism $\varphi$ extends the directed cover morphism $\psi$, that is $\varphi|_{G'} = \psi$.
	\end{enumerate}
	$$ \xymatrix{
		H' \ar@{.>}^-{\rm{Exc}}[r] \ar@{.>}_{\varphi}[d] & G' \ar^-{\psi}[d] \\
		H \ar^-{\rm{Exc}}[r] & G.
	}$$
	The statement holds replacing covers by emulators.
\end{lemma}

\begin{proof}
	The construction is explicit. Let $\psi = (p,q):G' \to G$ and let $E^{0}$ be the set of loops of $H$. For each loop $e \in E^{0}$, create a loop $e_{x'}$ in $G'$ at each preimage $x' \in p^{-1}(s(e))$ and set $s(e_{x'}) = t(e_{x'}) = x'$. Set
	$$H' = \left(V_{G'}, E_{G'} \bigcup_{e \in E^{0}} \bigcup_{v' \in p^{-1}(s(e))} e_{v'}, s_{G'} \cup s, t_{G'} \cup t \right). $$
	Extend the directed cover morphism $\psi$ to $H'$ by sending each edge $e_{x'}$ to $e$. Rename the extended direct cover morphism $\varphi$.
	By construction, ${\rm{Exc}}(H') = G'$ and $\varphi|_{G'} = \psi$. This proves $(1)$ and $(2)$. The construction is valid for directed emulators.
\end{proof}

There is a dual notion to the definition of a directed emulator. The original
definition distinguishes the outgoing edges. One could distinguish the incoming edges instead. Hence, parallel to Lemma \ref{lem:emu-surjects-onto-outedges}, we consider the set
${\rm{InE}}(y) = \{ e \in E \ | \ t(e) = y \}$ of incoming edges at $e$ and
require that the map ${\rm{InE}}(y') \to {\rm{InE}}(y)$ induced by the epimorphism be surjective. This defines a new notion of directed emulator, which we call {\emph{incoming directed emulator}}, while the original notion is
renamed an {\emph{outgoing directed emulator}}. If this does not seem to
cause confusion, we shall keep the terminology of directed emulator for an outgoing directed emulator. In contexts where we need to be more precise, we shall use the full terminology. A {\emph{bidirected emulator}} is a directed graph morphism that is both an incoming directed emulator and an outgoing directed emulator.

\begin{lemma} \label{lem:directed_emulators_and_opp}
	A directed graph morphism $\phi: G \to H$ is an incoming directed emulator if and only if $\phi^{op}: H^{\rm{op}} \to G^{\rm{op}}$ is a outgoing directed emulator.
\end{lemma}

\subsection{The case of undirected graphs} \label{subsec:undirected_emulators}

In this section, we discuss the relationship between directed and undirected graph emulators. First, we recall the original definition of an undirected graph emulator introduced by M.~R. Fellows in his PhD thesis in 1985 (see \cite{FL88}).\\

\noindent{\textbf{Fellows' definition}}. 
Let $G$ be an undirected graph. We say that an undirected graph
$G'$ is an {\emph{emulator}} of $G$ if there is a graph epimorphism $(p, q) : G' \to G$  such that for any edge $e\in E_G$ with $\partial(e)=\{ x, y \} $ and any $x' \in V_{G'}$ such that $p(x') = x$, there is an edge $e' \in E_{G'}$ such that $q(e') = e$. Again, we say that it is a cover when the edge $e'$ is uniquely defined. 

Note that the functor $U$ does not send directed emulators to emulators (and neither directed covers to covers). For instance, on the left, we have a directed emulator. But, on the right, node $u_1$ has only one adjacent edge where $w_1$ has two.
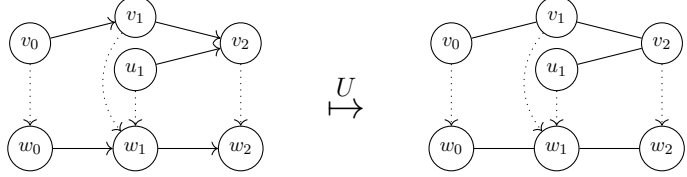
\begin{center}
	\begin{tikzpicture}[->,transform shape, scale=0.7, baseline=-2mm]]
	\node[etat](G0) at (0,0) {$v_0$};
	\node[etat] (G1) at (2,0.5) {$v_1$};
	\node[etat] (G1p) at (2,-0.5) {$u_1$};
	\node[etat](G2) at (4,0) {$v_2$};
	\node[etat](H0) at (0,-2) {$w_0$};
	\node[etat] (H1) at (2,-2) {$w_1$};
	\node[etat](H2) at (4,-2) {$w_2$};
	\path [] (G0) edge node [above=1mm] {} (G1)
	(G1p) edge node{} (G2)
	(G1) edge node [above=1mm] {} (G2)
	(H0) edge node{} (H1)
	(H1) edge node{} (H2);
	\draw (G0) [dotted] -- (H0);
	\draw (G1p) [dotted] -- (H1);
	\draw (G2) [dotted] -- (H2);
	\draw (G1) edge [dotted, bend right=40] (H1);
	\node[etat](K0) at (8,0) {$v_0$};
	\node[etat] (K1) at (10,0.5) {$v_1$};
	\node[etat] (K1p) at (10,-0.5) {$u_1$};
	\node[etat](K2) at (12,0) {$v_2$};
	\node[etat](L0) at (8,-2) {$w_0$};
	\node[etat] (L1) at (10,-2) {$w_1$};
	\node[etat](L2) at (12,-2) {$w_2$};
	\path [-] (K0) edge node [above=1mm] {} (K1)
	(K1p) edge node{} (K2)
	(K1) edge node [above=1mm] {} (K2)
	(L0) edge node{} (L1)
	(L1) edge node{} (L2);
	\draw (K0) [dotted] -- (L0);
	\draw (K1p) [dotted] -- (L1);
	\draw (K2) [dotted] -- (L2);
	\draw (K1) edge [dotted, bend right=40] (L1);
	\node (F) at (6,-1) {\huge $\stackrel{U}{\mapsto}$};
	\end{tikzpicture}
\end{center}

\begin{definition}
	The {\emph{bidirection}} of an undirected graph $G$ is the directed graph $\double{G} = (V, E, s, t)$ defined by $V = V_{G}$ and 
	\begin{eqnarray*}
	E = \{(e,x,y) \mid e \in E_G \text{ and } \partial_G(e) = \{x,y\} \},\\s(e,x,y) = x\ {\rm{ and}}\ t(e,x,y) = y. \end{eqnarray*}
\end{definition}

Notice that nonloop edges in $G$  are duplicated in $\double{G}$ whereas the set $ \{ \partial{e} = \{x\} \ | \ e \in E_{G} \}$ of loops in $G$ are in one-one correspondence with the set $\{ (e,x,x) \ | \ e \in E_{G}, \ \partial_{G}(e) = \{ x \} \}$ of loops in $\double{G}$.

Let $\phi = (p,q): G\to H$ be a graph morphism. We define $\double{\phi} = (p,q')$ with $q'((e,x,y)) = (q(e),p(x),p(y))$. It is clear that $\Delta_{\double{H}} (q'(e,x,y)) = \Delta_{\double{H}} (q(e),p(x),p(y)) = (p(x),p(y)) = p^{\times 2}(\Delta_{\double{G}}(e))$.

\begin{lemma} \label{lem:double_is_right_adjoint_to_U}
	The assignment $\double{(-)}$ yields a functor $\G \to \DG$ that is right adjoint to the functor $U$.
\end{lemma}

\begin{proof}
	For the statement about adjoints, let $G$ be a directed graph and let $H$ be an undirected graph. We define a map $\Phi_{G,H}: {\rm{Hom}}_{\DG}(G, \double{H}) \to {\rm{Hom}}_{\G}(U(G), H)$ as follows. Let $(p,q):G \to \double{H}$ be a directed graph morphism. Let $\pi_{3}^{1}:E_{\double{H}} \to E_{H}$, $\pi_{3}^{1}(e,x,y) = e$. Then set $$ \Phi_{G,H}(p,q) = (p, \pi_{3}^{1} \circ q):U(G) \to H.$$ 
	Conversely, given a morphism $(p',q'):U(G) \to H$, define a map $\Psi_{G,H}(p',q'):G \to \double{H}$ by $$\Psi_{G,H}(p',q') = (p',q''), \ {\hbox{with}}\ q''(e) = (q'(e),p'(s_G(e)),p'(t_G(e))).$$
	The maps $\Phi_{G,H}$ and $\Psi_{G,H}$ are inverse of each other, hence $\Phi_{G,H}:{\rm{Hom}}_{\DG}(G, \double{H}) \to {\rm{Hom}}_{\G}(U(G), H)$ is bijective. Naturality of the isomorphism $\Phi$ should be clear.
\end{proof}

\begin{lemma}\label{lem:double_respect_emulators}
	The morphism $\phi = (p,q):G\to H$ is an emulator if and only if $\double{\phi}:\double{G}\to\double{H}$ is a directed emulator.
\end{lemma}

\begin{remark} \label{rem:directed_on_double_is_bidirected}
	Note that $\double{\phi}:\double{G}\to\double{H}$ is a directed emulator if and only if it is a bidirected emulator. (Proof: use the natural isomorphism $(\double{G})^{\rm{op}} \simeq \double{G}$ for any $G$ and Lemma~\ref{lem:directed_emulators_and_opp}.)
\end{remark}

\begin{proof} Suppose $\phi$ is an emulator. Since $p$ is surjective, we only have to verify the outgoing edge lifting property. Suppose now $(e,x,y) \in \double{H}$ and $p(x') = x$ with $x' \in \double{G}$. Since $\phi$ is an emulator, there is an edge $e' \in E_G$ such that $q(e') = e$ and $\partial_G(e') = \{x',y'\}$ for some $y' \in V_G$. We have $q'(e',x',y') = (q(e'),p(x'),p(y')) = (e,x,p(y'))$. Do we have $p(y') = y$? Since $\Delta_{H}(q(e)) = p^{\otimes 2}(e) = \{p(x'), p(y')\} = \{x,y\}$, either $x = p(x') = p(y') = y$ or $p(x') \neq p(y')$ and since $p(x') = x$, we have $p(y') = y$.
	
	Suppose now that $\double{\phi} = (p,q')$ is a directed emulator. Again, since $p$ is surjective, we only have to verify the edge lifting property. 	Let $e \in E_{H}$ with $\partial_{H}(e) = \{x,y\}$ and $x' \in V_{G}$ with $p(x') = x$. Then, $(e,x,y) \in E_\double{H}$ and $x'\in V_G$ leads to an edge $(e',x',y') \in E_\double{G}$ with $q'(e',x',y') = (e,x,y)$. By definition of $q'$, that means that $q(e') = e$.  
\end{proof}

\begin{lemma} \label{lem:double_almost_adjoint_to_U}
	Suppose that $\phi : G \to \double{H}$ is a directed emulator, then, there is an emulator $\phi' : G' \to H$ with $g(G') = g(G)$. 
\end{lemma}

\begin{proof}
	Since $U$ is left adjoint to $\double{(-)}$, Lemma \ref{lem:double_is_right_adjoint_to_U} provides an undirected graph morphism $\phi':U(G) \to H$. It is clear that $g(\double{U(G)}) = g(G)$. It remains to check that $\phi'$ is an emulator morphism. For this, let $\phi = (p,q)$ so that $\phi'=(p,\pi^1_3 \circ q)$ with $\pi_3^1(e,x,y) = e$. Let $e\in E_H$ such that $\partial_H(e) = \{x,y\}$ and $x' \in V_{U(G)} = V_G$ such that $p(x') = x$. Then, $\arete{x}{(e,x,y)}{y} \in \double{H}$ and $p(x') = x$ implies the existence of $e'\in E_G$ such that $q(e') = (e,x,y)$ with $s_G(e') = x'$. We have $\pi_1^3\circ q(e') = e$ and $p(e') = x'$ as expected.  
\end{proof}

\begin{remark}
	The converse of Lemma \ref{lem:double_almost_adjoint_to_U} does not hold: the directed graph morphism $G \to \double{H}$ induced by an emulator morphism $U(G) \to H$ is not a directed emulator morphism in general. A counterexample is provided by $U(G) = H =
	\begin{tikzpicture}[-,transform shape,scale=0.6, baseline=-1mm]]
	\node[etat] (H1) at (0,0){};
	\node[etat] (H2) at (2,0) {};
	\node[etat] (H3) at (4,0){};
	\path [] (H1) edge node []{} (H2)
	(H2) edge node []{} (H3);
	\end{tikzpicture}$. 
\end{remark}

\begin{remark} \label{rem:does_not_work_for_covers}
	The isomorphism ${\rm{Hom}}_{\DG}(G, \double{H}) \to {\rm{Hom}}_{\G}(U(G), H)$ provided by Lemma~\ref{lem:double_is_right_adjoint_to_U} is shown in the proof of Lemma~\ref{lem:double_almost_adjoint_to_U} to restrict to an isomorphism from directed emulators to undirected emulators. However, it does not restrict to an isomorphism from \emph{directed covers} to \emph{undirected covers}. 
A counterexample is provided by $G = \begin{tikzpicture}[transform shape, scale=0.6,baseline=-1mm]
	\node[etat] (H1) at (0,0){};
	\node[etat] (H2) at (2,0){};
	\path[->] (H1) edge[bend left] node []{} (H2)
	(H2) edge[bend left] node []{} (H1);
	\end{tikzpicture}$ and $H = \begin{tikzpicture}[transform shape, scale=0.6,baseline=-1mm]
	\node[etat] (H1) at (0,0){};
	\node[etat] (H2) at (2,0){};
	\path[-] (H1) edge node []{} (H2);
	\end{tikzpicture}$.

\end{remark}

\begin{wrapfigure}{r}[-2mm]{0.15\textwidth}
\begin{tikzpicture}[-, every loop/.style={}, transform shape, scale=0.5, baseline=-2mm]
 \tikzset{vertex/.style={circle,minimum size=12pt,inner sep=4pt}}
\node[vertex](G0) at (0,0) [shape=circle,draw=black] {};
\node[vertex] (G1) at (0,-2)[shape=circle,draw=black]  {};
\node[] (F) at (0,-3.75) {\Huge $\downarrow$};
\node[vertex](G2) at (0,-6) [shape=circle,draw=black] {};
\path[-, line width = 0.5mm] (G0) edge node [] {} (G1);
\path[-, line width = 0.5mm] (G2) edge [loop, out = 220, in = 140, distance = 2cm] node {} (G2);
\end{tikzpicture}
\end{wrapfigure}

A (complete) \emph{direction} $\vec{G}$ of an undirected graph $G$ is a subgraph $H$ of $\double{G}$ such that $U(H)$ is isomorphic to $G$. It is easy to see that given a (complete) direction $\vec{G}$ of $G$ and a graph morphism $G' \to G$, there exists a (complete) direction $\vec{G'}$ of $G'$ inducing a directed graph morphism $\vec{G'} \to \vec{G}$. This observation does not hold in the categories of emulators and directed emulators. For instance, the picture opposite  represents an undirected emulator (actually even an undirected cover). However, there is no choice of directions that turns it into a directed emulator: no matter which directions are chosen, one of the two vertices in the preimage of the vertex fails to satisfy the lifting outgoing edge property.

\begin{lemma}[Lifting lemma] \label{lem:from_emulator_to_directed_emulator}
	Let $\phi: G' \to G$ be an undirected graph emulator morphism onto a loopless graph $G$. Given a direction $\vec{G}$ of $G$, there is a direction $\vec{G'}$ of $G'$ and a directed emulator $\phi':\vec{G'} \to \vec{G}$.
\end{lemma}

\begin{proof} By Lemma~\ref{lem:double_respect_emulators},  $\double{\phi}: \double{G'} \to \double{G}$ is a directed emulator morphism. Then, apply Corollary~\ref{sub-graph-preserve} to the subgraph $\vec{G}$ of $\double{G}$: there is a subgraph $H = \double{\phi}^{-1}(\vec{G})$ of $\double{G'}$ such that $\double{\phi}$ restricts to a directed emulator $H \to \vec{G}$. If $G$ is loopless, then $H = \double{\phi}^{-1}(\vec{G})$ is a direction of $G'$. 
\end{proof}

\begin{remark}
If $G$ has loops then $H$ may not be a direction of $G$, as the example in the previous paragraph shows.
\end{remark}

The undirected version of the extraction of a directed cover from a directed emulator  (Lemma \ref{lem:directed-emulator-contains-directed-cover}) fails to hold. The following picture depicts a $4$-vertex undirected emulator of a simple undirected $3$-vertex graph that does not contain a $4$-vertex cover as a subgraph.

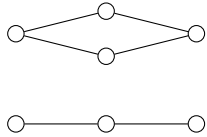
\begin{center}
	\begin{tikzpicture}[-,transform shape,scale=0.6, baseline=-2mm]]
	\node[etat] (H1) at (0,0){};
	\node[etat] (H2) at (4,0) {};
	\node[etat] (H3) at (2,0.5){};
	\node[etat] (H4) at (2,-0.5){};
	\node[etat] (G1) at (0,-2){};
	\node[etat] (G2) at (2,-2){};
	\node[etat] (G3) at (4,-2){};
	\path [] (H1) edge node []{} (H3)
	(H3) edge node []{} (H2);
	\path[] (H1) edge node []{} (H4)
	(H4) edge node[]{} (H2);
	\path[] (G1) edge node []{} (G2)
	(G2) edge node []{} (G3);
	\end{tikzpicture}
\end{center}
It is left to the reader to verify on this example that any direction of
the bottom graph yields a directed emulator on the top by lifting directions
(Lemma \ref{lem:from_emulator_to_directed_emulator}) and that from it a directed cover with all the original vertices can
be extracted (Lemma \ref{lem:directed-emulator-contains-directed-cover}).

In other words, roughly speaking, it is much easier to find directed emulators compared to emulators (or bidirected emulators).

For the remainder of the paragraph, let us define, for a directed (resp. undirected) graph $G$, 
$$ g_{\rm{cover}}(G) = \min \{ g(G') \ | \ G'\ {\rm{covers}}\ G\}, \ \ g_{\rm{em}}(G) = \min \{ g(G') \ | \ G'\ {\rm{emulates}}\ G\}.$$

By Proposition~\ref{prop:directed_is_cover} above, $g_{\rm{cover}}(G) = g_{\rm{em}}(G)$ for any {\emph{directed}} graph $G$. One can ask whether this equality also holds for undirected graphs. In 1999, 
P.~Hlin\v{e}n\`y \cite{Hl99} found an
emulator such that $$ g_{\rm{em}}(G) \leq 3 < g_{\rm{cover}}(G).$$
In particular, there exist undirected emulators of $G$ that do not contain covers of $G$.

M.~Fellows proposed the following conjecture for undirected graphs in 1985.

\begin{conj} \label{conj:fellows}
	A connected graph has a finite planar emulator if and only if it has a finite planar cover.
\end{conj}

The conjecture was proved false in 2009 by Y.~Rieck and Y.~Yamashita who found a counterexample \cite{RY09}. Therefore, there are graphs $G$ such that $g_{\rm{em}}(G) = 0$ and $g_{\rm{cover}}(G) \geq 1$.

The following conjecture is still open.

\begin{conj}[Negami]
A connected graph has finite planar cover if and only if its embeds into the projective plane.
\end{conj}

\section{An automatic description of directed emulators and covers}

\newcommand{\sursim}[1]{{#1}{/\!\!\sim}}


\subsection{Automatic relations}

In this paragraph we give an alternative description of directed emulators and covers. Let $G = (V,E,s,t)$ be a directed graph.

\begin{definition} \label{def:relations_on_graph}
A pair of equivalence relations $(\sim_V,\sim_E)$ on respectively $V$ and $E$ is said to be \emph{automatic} when for all edges $e, e'$, and any vertex $x'$:
\begin{itemize}
\item[(i)]$e \sim_{E} e'\ \Longrightarrow\ s(e) \sim_V s(e')\ \wedge\ t(e) \sim_V t(e')$;
\item [(ii)] $x' \sim_{V} s(e) \ \Longrightarrow\ \exists e' \in E : e' \sim_{E} e \ \wedge \ s(e') = x'$.
\end{itemize}
\end{definition}
Clause (i) is next called \emph{compatibility} (of $\sim_E$ with respect to $\sim_V$). Clause (ii) is the \emph{bisimilarity} of $\sim_V$ with respect to $\sim_E$. 

Given an equivalence relation $\sim_V$ on vertices, let $e \sim_E e' \Longleftrightarrow (s(e) \sim_V s(e')) \wedge (t(e) \sim_V t(e'))$. Note that $\sim_E$ is compatible with $\sim_V$. If $\sim_V$ is bisimilar with respect to $\sim_E$, the pair $(\sim_V,\sim_E)$ forms an automatic relation. Such a relation is said to be vertex-induced.  

Given an equivalence relation $\sim$ on $V$, we denote $[x]_\sim = \{ z \in V \mid z \sim x\}$ be the equivalence class of $x \in V$.  If the context is clear, we write $[x]$ for $[x]_\sim$. We let $\sursim{V} = \{ [v] \mid v \in V\}$. In order to avoid notation overload we remove subscripts from $\sim_{V}$ and $\sim_{E}$ whenever this does not seem to cause confusion.

\begin{lemma} \label{lem:automatic_merge_well_defined}
An automatic relation $\sim$ on $G$ induces a new directed graph $\sursim{G} = (\sursim{V_G}, \sursim{E_G}, \sursim{s},\sursim{t})$ together with a morphism $(x\mapsto [x]_{\sim_V}, e \mapsto [e]_{\sim_E}) : G \to \sursim{G}$ called the canonical morphism (with respect to $\sim$). It is a directed emulator.
\end{lemma}

\begin{proof}It is clear that $[-]_\sim$ is an epimorphim. Let us verify the outgoing edge lifting property. Let $v \in V$ and let
 $\arete{[v]}{[e]}{[w]}$ be an edge in $\sursim{G}$. Since $[-]_{\sim}$ is onto, there is an edge $\arete{x}{e}{y} \in E$ such that $v \sim x$. By bisimilarity, there is an edge $\arete{v}{e'}{v'}$ in $G$ such that $e' \sim e$ so that $[e'] = [e]$. 
\end{proof}

We say that $\sursim{G}$ is an {\emph{automatic quotient}} of $G$.

\begin{example} Given any directed graph $G$, the identity relation $\sim$ on both vertices and edges is automatic and $\sursim{G} \simeq G$.
\end{example}

\begin{lemma} \label{lem:max_auto_implies_simple}
If $\sim$ is vertex-induced automatic then $\sursim{G}$ is simple.
\end{lemma}
Let $1_V$ be the equality on the set $V$ of vertices of some graph $G$. We denote $\sim_G$ its vertex-induced automatic relation. 

\begin{corolem} $G/{\sim_G} \simeq R(G)$. 
\end{corolem}

\begin{prop}\label{pr:presentation_of_covers}
The morphism associated to a pair of automatic relations $(\sim_V,\sim_E)$ is a cover if and only if (iii): $(e\sim_E e' \ \wedge\ e \neq e')  \Longrightarrow \ s(e) \neq s(e')$.
\end{prop}

\begin{proof} Clause (iii) is a reformulation of the injectivity of the projection.
\end{proof}

In other words, adding clause (iii) in Definition~\ref{def:relations_on_graph} lead to a description of covers. 

\begin{prop} \label{lem:diem_implies_automatic}
Let $(p,q):G \to H$ be a directed emulator morphism. The relation $\sim$ defined on $G$ by $v \sim_{V} w$ if $p(v) = p(w)$ and by $e \sim_{E} e'$ if $q(e) = q(e')$ is an automatic equivalence relation on $G$. Furthermore, $H$ is simple if and only if $\sim$ is vertex-induced.
\end{prop}

\begin{proof}
The relation $\sim$ is clearly an equivalence relation. Let $x\sim y$ be two equivalent vertices. Let $\arete{x}{e}{x'} \in E_{G}$ be an edge. Applying $q$ yields an edge $\arete{p(x)}{q(e)}{p(x')} \in E_{H}$. By the outgoing edge lifting property, there is an edge $\arete{y}{e'}{y'}$. We have $q(e') = q(e)$ so $p(y') = p(x')$, i.e. $x' \sim_{V} y'$. Hence $\sim_{V}$ has the bisimilarity property with respect to $\sim_E$. Compatibility is a direct consequence of adjacency preservation of $q$ with respect to $p$. 

For the second statement, the quotient by a vertex-induced automatic relation is simple by Lemma~\ref{lem:max_auto_implies_simple}. Conversely, suppose that $H$ is simple. Consider a pair of edges $e, e' \in G$ such that $s(e) \sim_{V} s(e')$ and $t(e) \sim_{V} t(e')$, i.e.
$\Delta(q(e)) = \Delta(q(e'))$. Since $H$ is simple, $q(e) = q(e')$. Thus $e \sim_{E} e'$. The conclusion follows.
\end{proof}

\begin{definition}
The relation defined by Prop.~\ref{lem:diem_implies_automatic} is the {\emph{canonical automatic relation}} associated to the directed emulator morphism $G \to H$.
\end{definition}

\begin{prop}[composition of automatic relations] \label{prop:auto_rels_compose}
Let $\sim_1$ be an automatic relation on a directed graph $G$ and $\sim_2$ an automatic relation on $G/\sim_1$. Define the relation $\sim$ on $G$ by $e \sim e'$ if $[e]_{\sim_1} \sim_{2} [e]_{\sim_1}$ and by $v \sim v'$ if $[v]_{\sim_1} \sim_{e} [v']_{\sim_1}$. The relation $\sim$ is automatic and verifies $[-]_{\sim} = [-]_{\sim_2} \circ [-]_{\sim_1}$.
\end{prop}

\begin{proof}
By definition, $[e]_{\sim} = [[e]_{\sim_1}]_{\sim_{2}}$ and $[v]_{\sim} = [[v]_{\sim_1}]_{\sim_2}$ for any edge $e$ and any vertex $v$ in $G$. The fact that the relation is automatic follows from the definition.
\end{proof}

\begin{remark} \label{rem:iso_induces_identity_for_auto_rel}
Recall that an isomorphism $\phi:G \to H$ of directed graphs is a directed cover (Example~\ref{ex:iso_is_directed_emulator}). The canonical automatic relation $\sim$ associated to any isomorphism $\phi$ is the identity relation (defined by the identity on the vertices and the identity on the edges).
\end{remark}

\begin{theorem}\label{th:unique_automatic_decomposition} Any directed emulator morphism $\phi: G \to H$ splits in a unique way as $\phi = \iota  \circ [-]_\sim$ where $\sim$ an automatic relation on $G$ and $\iota$ is an isomorphism.
\end{theorem}

\begin{proof}  Let $\sim$ be the canonical automatic relation related to $\phi =(p,q)$. 
We define $\iota =(p',q')$ with $p' : [v]_\sim \mapsto p(v)$ and $q' : [e]_\sim \mapsto q(e)$. Let us verify that the definition is correct. For two edges $e_1$ and $e_2$, we have $[e_1]_\sim = [e_2]_\sim$ iff $e_1 \sim e_2$ iff $q(e_1) = q(e_2)$. Actually, this proves that $q'$ is injective. Moreover $q'$ is  surjective since $q$ is surjective itself. Similarly $p'$ is bijective. Thus $\iota$ is an isomorphism. 
Finally, by construction, $\iota \circ [-]_\sim =\phi$. 
	
Suppose there is an other decomposition $\phi = \iota' \circ [-]_{\sim'}$. Suppose that $[v]_\sim \neq [v]_{\sim'}$. Thus, either there is some $u \sim v$ and $u \not\sim' v$ or $u \not\sim v$ and $u \sim' v$. In the first case, $\iota'([u]_{\sim'}) = p(u) = p(v) = \iota'([v]_{\sim'})$ contradicts the fact that $\iota'$ is an isomorphism. Otherwise, $u \sim' v$ but not $u \sim v$. Thus, $p(v) \neq p(u)$. However, $u \sim' v$ means $[u]_{\sim'} = [v]_{\sim'}$ leading to $p(u) = \iota'([u]_{\sim'}) = \iota'([v]_{\sim'}) = p(v)$. Again, a contradiction. The same argument applies to edges.  
\end{proof}

\begin{example}
In the particular case when $\phi:G \to H$ is an isomorphism of directed graphs, $[-]_{\sim} = 1_{G}$ and $\iota = \phi$.
\end{example}

As suggested by the previous example, directed graph isomorphisms and canonical epimorphisms associated to automatic relations appear to be ``orthogonal'' notions. 

\begin{prop}
The pair $({\mathscr{A}},{\mathscr{I}})$, that consists of the class of canonical epimorphisms associated to automatic relations on one side and the class of directed graph isomorphims on the other, is a factorization system in the category $\Emu$.
\end{prop}

\begin{proof}
 The class of canonical epimorphisms associated to automatic relations is closed under composition according to Proposition~\ref{prop:auto_rels_compose}, as well as the class of directed graph isomorphisms. Theorem~\ref{th:unique_automatic_decomposition} provides a decomposition $\phi = \iota \circ [-]_{\sim}$ for any directed emulator morphism $\phi$ where $\iota \in {\mathscr{I}}$ and $[-]_{\sim} \in {\mathscr{A}}$. It is easy to check that the decomposition is functorial. 
\end{proof}

We are on the way to Theorem~\ref{th:terminal_in_slice} (see end of \S \ref{subsec:partial_order_on auto_rel} below) as this is done by T.~Colcombet and D.~Petri\c{s}an in \cite{ColcombetP19}.

\subsection{MN-recursive relations} In this paragraph, the definition of an automatic relation is refined by labelling the edges. The first observation is that a directed graph endowed with an automatic relation induces naturally a semi-automaton. We then describe automatic relations in terms of a certain kind of recursive relations reminiscent of the Myhill-Nerode relation.

\begin{definition} \label{def:canonical_semi_auto}
A directed graph $G$ endowed with an automatic relation $\sim$ induces a semi-automaton $(G,\Sigma, \ell)$ defined by
$$ \Sigma = E_{G}/\!\sim_{E}\ {\hbox{and}}\ \ell:E\to E/\!\sim_{E}, e \mapsto [e].$$
This transition system, denoted $\A_{\sim}$ is the {\emph{canonical semi-automaton associated to the automatic relation}}\ $\sim$. 
\end{definition}

\begin{remark}
The canonical semi-automaton $\A_{\sim}$ depends only on the equivalence relation $\sim_{E}$ on the set $E_{G}$ of edges. However we are mainly interested in refinements of automatic relations in semi-automata. See for instance Lemma~\ref{lem:label-aut_is_aut} and Definition~\ref{def:MN-recursive relation}.
\end{remark}

\begin{remark}
The (important) special case when the automatic relation is the identity on vertices and edges (tautological semi-automata) is studied in \ref{sub:graph-to-semi-auto}.
\end{remark}

\begin{definition}
A pair of equivalence relations $\sim$ on a semi-automaton $\A$ is said to be {\emph{automatic}} if it is an automatic relation on its underlying graph $G_{\A}$ and $e\sim_E e'$ implies $\ell(e) = \ell(e')$ for all edges $e, e'$.
\end{definition}

The following observation should be clear.

\begin{lemma} \label{lem:label-aut_is_aut}
A pair of relations $(\sim_{E}, \sim_{V})$ on a directed graph $G$ is automatic if and only if it is automatic for the canonical semi-automaton $\A_{\sim}$.
\end{lemma}

\begin{proof}
Suppose that $(\sim_{V}, \sim_{E})$ is automatic. Labelling the set $E$ of edges by $\ell:E \to E/{\sim_{E}}, \ e \mapsto [e]$ makes the relation automatic for the semi-automaton. 
The converse is trivial.
\end{proof}

\newcommand{\mnr}[2]{{\triple{#1}{\sim}{#2}}}

For the following definitions, we consider a fixed semi-automaton $\A$ and a family ${\mathscr{F}} = \{ F_{i} \}_{i \in I}$ of non-empty pairwise disjoint subsets of the set $V_{\A}$ of vertices of $\A$.

\begin{definition} \label{def:labelled_equiv}
We define on $V_{\A}$ the equivalence relation $x \triple{0}{\sim}{{\mathscr{F}}} y \Longleftrightarrow \exists i : (x\in F_i \Leftrightarrow y \in F_i)$.  Given $n \geq 1$, let $\triple{n}{\sim}{\mathscr{F}}$ be the least equivalence relation such that $x\triple{n}{\sim}{\mathscr{F}} y$ holds whenever the two following conditions are satisfied:
\begin{enumerate}[(1)]
\item $x \triple{n-1}{\sim}{\mathscr{F}} y$;
\item for each edge $\arete{x}{e}{x'}$, there is an edge $\arete{y}{f}{y'}$ such that $\ell(f) = \ell(e)$ and $x' \triple{n-1}{\sim}{\mathscr{F}} y'$.
\end{enumerate}
\end{definition}

It follows from the definition that there exists $n \leq |V_{\A}|$ such that $\triple{n}{\sim}{\mathscr{F}}\, = \triple{n+1}{\sim}{{\mathscr{F}}}$. It follows from that we can drop the underscript $n$ from the notation as long as $n \geq |V_{\A}|$. We shall simply write $u \overset{\mathscr{F}}{\sim} v$ without further comment. 

\begin{remark} \label{rem:partition}
A family of disjoint non-empty subsets of $V$ naturally induces a partition of $V$ as follows.
Let ${\mathscr{F}} = \{ F_{i}, \ i \in I\}$ be any family of disjoint non-empty subsets of $V$. Let ${\mathscr{F}}'$ be the partition of $V$ obtained from ${\mathscr{F}}$ by adding the complement in $V$ of the union of all subsets $F_{i}$, i.e., let ${\mathscr{F}}' = \mathscr{F} \cup \{ \left( \cup_{i \in I} F_{i} \right)^{\complement} \}$. Then $\overset{\mathscr{F}}{\sim}\, =\, \overset{\mathscr{F}'}{\sim}$. That is, both the family ${\mathscr{F}}$ and 
 \emph{the partition ${\mathscr{F}}'$ induced by} ${\mathscr{F}}$ yield the same equivalence relation.
\end{remark}

\begin{example}[Myhill-Nerode equivalence relation]
Suppose that ${\mathscr{F}}$ consists of one unique subset $F$ of states (or if one thinks in terms of partition, of $F$ and $F^{\complement}$): ${\mathscr{F}} = \{ F \}$. In this case we denote the relation $\overset{\{ F \}}{\sim}$ simply by $\overset{F}{\sim}$. Set $F$ to be the subset of final states. The relation $\overset{\{ F \}}{\sim}$ is nothing but the Myhill-Nerode equivalence relation on the set of states, expressed algorithmically in order to recursively build all classes of equivalent states.
\end{example}

\begin{lemma} \label{lem:finer_partition}
If the partition ${\mathscr{G}}$ is finer than ${\mathscr{F}}$ then $\overset{\mathscr{G}}{\sim}\, \subseteq\, \overset{\mathscr{F}}{\sim}$.
\end{lemma}


\begin{definition}[MN-recursive relation in a semi-automaton] \label{def:MN-recursive relation}
We define a relation $MN(\A,\mathscr{F}) = (\overset{\mathscr{F}}{\sim}_{V},\overset{\mathscr{F}}{\sim}_E)$ on $\A$ by setting
$$e \overset{\mathscr{F}}{\sim}_E e' \Longleftrightarrow s(e) \overset{\mathscr{F}}{\sim}_{V} s(e') \ \wedge \ t(e) \overset{\mathscr{F}}{\sim}_{V} t(e') \ \wedge \ \ell(e) = \ell(e').$$ Such a relation is said to be \emph{MN-recursive}. 
\end{definition}

\begin{lemma} \label{lem:finer_partition_II}
If the partition ${\mathscr{G}}$ is finer than ${\mathscr{F}}$ then ${\rm{MN}}(\A, {\mathscr{G}}) \subseteq {\rm{MN}}(\A, {\mathscr{F}})$.
\end{lemma}

\begin{proof}
Follows from Lemma \ref{lem:finer_partition} and the definition of an MN-recursive relation.
\end{proof}

\begin{lemma} \label{lem:MN-implies-automatic} An MN-recursive relation is automatic.
\end{lemma}
\begin{proof}
By construction, $\sim_E$ is compatible with $\mnr{}{\mathscr{F}}$ and the labelling. Let us check $\mnr{}{\mathscr{F}}$ is bisimilar with respect to $\sim_E$. Suppose $v \mnr{}{\mathscr{F}} w$ and $\arete{v}{e}{v'} \in G_\A$. By definition, there is an edge $\arete{w}{f}{w'}$ such that $v'\mnr{}{\mathscr{F}}w'$ and $\ell(e) = \ell(f)$. But then, by definition, $e \sim_E f$ is as expected.
\end{proof}

The next two observations aim at identifying MN-recursive relations. The first observation says that in order to prove that an MN-recursive relation on a canonical semi-automaton coincides with an automatic relation, it suffices to show that they coincide on the set of vertices. The second observation gives a partial criterion for this.

\begin{lemma} \label{lem:vertices_suffice}
Let $\sim\, = (\sim_{V}, \sim_{E})$ be an automatic relation on a directed graph $G$ and let ${\mathscr{F}}$ be any partition of $V_{G}$. Then, ${\rm{MN}}(\A_{\sim}, {\mathscr{F}}) =\, \sim$ if and only if $\overset{\mathscr{F}}{\sim}_{V}\, = \, \sim_{V}$.
\end{lemma}

\begin{proof}
Trivially if the two relations coincide, they coincide on the set of vertices. Conversely suppose that $\sim_{V}\, = \, \overset{\mathscr{F}}{\sim}_{V}$. The equality of the relations on the edges $\overset{\mathscr{F}}{\sim}_{E} = \sim_{E}$ follows from the series of equivalences
\begin{align*}
e\overset{\mathscr{F}}{\sim}_{E} e' & \Longleftrightarrow s(e) \overset{\mathscr{F}}{\sim}_{V} s(e') \ \wedge \ t(e) \overset{\mathscr{F}}{\sim}_{V} t(e') \ \wedge \ \ell(e) = \ell(e') \\
& \Longleftrightarrow s(e) {\sim}_{V}\, s(e') \ \wedge \ t(e) {\sim}_{V}\, t(e') \ \wedge \ e \sim_{E} e'\\
& \Longleftrightarrow e \sim_{E} e'.
\end{align*}
\end{proof}

\begin{remark} \label{rem:inclusions}
The lemma is true, more generally, if in the statement we replace the equality $=$ between the relations by $\subseteq$ or $\supseteq$. The proof is mutatis mutandis the same.
\end{remark}

\begin{lemma} \label{lem:vertices_suffice_at_level_0}
Let $\sim\, = (\sim_{V}, \sim_{E})$ be an automatic relation on a directed graph $G$ and let ${\mathscr{F}}$ be any partition of $V_{G}$. Then $\sim\, \subseteq {\rm{MN}}(\A_{\sim}, {\mathscr{F}})$ if and only if $\sim_{V}\, \subseteq\,\, \underset{\!\!\!\! 0}{\overset{\mathscr{F}}{\sim}_{V}}$.
\end{lemma}

\begin{proof}
According to Remark~\ref{rem:inclusions}, $\sim\, \subseteq \,\, \overset{\mathscr{F}}{\sim}$ if and only if $\sim_{V}\, \subseteq\,\,{\overset{\mathscr{F}}{\sim}_{V}}$. Thus it is enough to prove that $\sim_{V}\, \subseteq\,\,{\overset{\mathscr{F}}{\sim}_{V}}$ if and only if $\sim_{V}\, \subseteq\,\, \underset{\!\!\!\! 0}{\overset{\mathscr{F}}{\sim}_{V}}$. The direct implication is trivial. Let us prove the converse. For simplicity, we drop the subscript $V$ from the notation whenever there should be no confusion. Suppose that $x \sim y$ implies $x \, \underset{0}{\overset{\mathscr{F}}{\sim}}\, y$ for any $x, y \in V_{G}$. Suppose that $x \sim y$. We shall prove that $x \underset{n}{\overset{\mathscr{F}}{\sim}} y$ for any $n$ by induction on $n \geq 0$. By assumption, the base case $n = 0$ holds. Suppose (induction hypothesis) that we have proved that $x \sim y$ implies that $x \underset{n-1}{\overset{\mathscr{F}}{\sim}} y$. Note that this implies that condition $(1)$ of Definition~\ref{def:labelled_equiv} holds. Let us verify that condition $(2)$ holds as well. Given an edge $\arete{x}{e}{x'}$, the bisimilarity property for $\sim_{V}$ with respect to $\sim_{E}$ 
yields an edge $\arete{y}{f}{y'}$ such that $e \sim_{E} f$.  The compatibility property for $\sim_{V}$ with respect to $\sim_{E}$ 
 then yields $x' = t(e) \sim_{V} t(f) = y'$. By the induction hypothesis, this implies that $x'\, \mnr{n-1}{\mathscr{F}}\, y'$.
Hence $x\, \mnr{n}{\mathscr{F}}\, y$ holds as claimed.
\end{proof}

It will be convenient to define, for any vertex $w \in V$, the subset $${\rm{Pr}}(w)  = \{ v \in V \ | \ {\hbox{there is a walk starting at}}\ v\ {\hbox{and ending at}}\ w \}.$$ We accept length $0$ walks so $w \in {\rm{Pr}}(w)$ for any vertex $w \in V_{G}$. The existence of a walk from $v$ to $w$ ($w$ is reachable from $v$) implies ${\rm{Pr}}(v) \subseteq {\rm{Pr}}(w)$. A subset $\{ v_{i}\}_{i \in I}$ of vertices in $V_{G}$ such that $V = \cup_{i \in I} {\rm{Pr}}(w_{i})$ will be called a {\emph{complete final system of the directed graph}} $G$. A complete final system $\{v_{i}\}_{i \in I}$ of 
$G$ is {\emph{minimal}} if no strict subfamily $\{ v_{j}\}_{j \in J}$, $J \subset I, J \not= I$, is a final system of $G$. Although we shall not use it in the sequel, note that the cardinality of a minimal complete final system is an invariant of the directed graph.

\begin{lemma}
Minimal complete final systems of $G$ have the same cardinality.
\end{lemma}

\begin{proof}
Let $S, T$ be two minimal complete final systems.  Let $s \in S$. By completeness of $T$, there is a walk starting at $s$ ending at some $t \in T$. By completeness of $S$, there is a walk starting at $t$ ending at some $s' \in S$. Now minimality implies $s = s'$ for otherwise since ${\rm{Pr}}(s) \subseteq {\rm{Pr}}(s')$, the smaller system $S - \{ s \}$ would still be complete. Thus $s \in S$ is paired to an element $t \in T$ such that ${\rm{Pr}}(s) = {\rm{Pr}}(t)$. Moreover there is no further walk from $s$ to some $t' \in T$, $t' \not= t$, for otherwise the smaller system $T - \{ t \}$ would still be complete. Therefore, the assignment that sends $s \in S$ to the unique $t \in T$ such that ${\rm{Pr}}(t) = {\rm{Pr}}(s)$ is a well-defined map $S \to T$. The symmetry of $S$ and $T$ implies that it is bijective.
\end{proof}

We are now ready to prove a converse to Lemma~\ref{lem:MN-implies-automatic}.

\begin{prop}[Automatic relations are MN-recursive]
\label{prop:automatic-implies-MN}
If a pair of relations $\sim\, = (\sim_{V}, \sim_{E})$ on a directed graph $G$ is automatic, then $\sim\, =  {\rm{MN}}(\A_{\sim}, {\mathscr{F}})$ for the canonical semi-automaton $\A_{\sim}$ associated to $\sim_{E}$ and for the family ${\mathscr{F}}$ that consists of the $\sim_{V}$-equivalence classes of the vertices of any complete final system of $G$. 
\end{prop}


\begin{proof}
Let $S$ be a minimal complete final system for $G$ and let ${\mathscr{F}}$ denote the corresponding partition of $V$. In accordance with Remark~\ref{rem:partition}, each element in ${\mathscr{F}}$ is $[s]_{\sim_{V}}$, $s \in S$, and possibly (if non empty) the subset $\left( \cup_{s \in S} [s]_{\sim_{V}} \right)^{\complement}$. 
In accordance to Lemma~\ref{lem:vertices_suffice}, it suffices to prove that $\sim_{V}\, =\, \overset{\mathscr{F}}{\sim}_{V}$. 

For simplicity, we drop the subscript $V$ from the notation. Let us prove first that $\sim_{V}\, \subseteq\, \overset{\mathscr{F}}{\sim}_{V}$. Suppose that $x \sim y$. Either $x \in [s]$ for some $s \in S$ (and then $x$ belongs to exactly one $[s] \in {\mathscr{F}}$) or $x \in \left( \cup_{s \in S} [s]_{\sim_{V}} \right)^{\complement}$ (and then $x$ also belongs to exactly one class in ${\mathscr{F}}$). Since $x \sim y$, it follows that $y$ has the same property as $x$. Hence $x \overset{\mathscr{F}}{\underset{0}{\sim}} y$. Then by Lemma~\ref{lem:vertices_suffice_at_level_0}, $x \overset{\mathscr{F}}{\sim} y$.

Conversely, assume that $x \triple{}{\sim}{\mathscr{F}} y$. We have to prove that $x \sim y$. By assumption $x \overset{\mathscr{F}}{\underset{0}{\sim}} y$, that is, either there is $s \in S$ such that $[x] = [y] = [s]$ (and then $x \sim y$) or $x, y \in (\cup_{s \in S} [s])^{\complement}$. 
Given a vertex $x$ in $G$, there is at least one walk from $x$ to (one of the vertices in) $[s]$ for some $s \in S$. That justifies to proceed by induction on the length of such a walk (of minimal length). By such an induction, we prove that $x\triple{}{\sim}{\mathscr{F}}y$ implies $x\sim y$.
The walk has length $0$ if and only if $[x] = [s] = [y]$, as we've just already observed. Otherwise, there is an edge $\arete{x}{e}{x'}$, where the vertex $x'$ is ``closer" to $[s]$ for some $s \in S$. Since the relation $\mnr{}{\mathscr{F}}$ is automatic, there is an edge $\arete{y}{f}{y'}$ with $e\stackrel{\mathscr{F}}{\sim}_E f$. Hence $\ell(e) = \ell(f)$ for the label $\ell$ of the canonical semi-automaton $\A_{\sim}$, thus $e \sim_E f$. So $x = s(e) \sim_V s(f) = y$ by compatibility of $\sim_E$ with respect to $\sim_V$. 
\end{proof}

\begin{theorem}\label{th:automatic_as_mn}
A relation on a directed graph $G$ is automatic if and only if it is MN-recursive with respect to some semi-automaton $\A = (G, \Sigma, \ell)$. More precisely, given a pair of relations $\sim = (\sim_V, \sim_E)$ on a directed graph $G$, the following assertions are equivalent:
\begin{enumerate}
\item[$(1)$] The relation $\sim$ on $G$ is automatic.
\item[$(2)$] The relation $\sim$ on the canonical semi-automaton $\A_{\sim}$ is automatic.
\item[$(3)$] $\sim\, = {\rm{MN}}(\A_{\sim}, {\mathscr{F}})$
where ${\mathscr{F}}$ is the family of $\sim_{V}$-equivalence classes of the vertices of any complete final system of $G$.
\item[$(4)$] $\sim\, = {\rm{MN}}(\A_{\sim}, {\mathscr{G}})$
where ${\mathscr{G}}$ is the partition of $V$ into $\sim_{V}$-equivalence classes.
\end{enumerate}
\end{theorem}

\begin{proof}
$(1) \Longleftrightarrow (2)$: Lemma \ref{lem:label-aut_is_aut}. $(2) \Longrightarrow (3)$: Prop.~ \ref{prop:automatic-implies-MN}. To end the proof, we have to prove that ${\rm{MN}}(\A_{\sim}, {\mathscr{G}}) = {\rm{MN}}(\A_{\sim}, {\mathscr{F}})$. Since ${\mathscr{G}}$ is finer than (the partition  induced by) ${\mathscr{F}}$, Lemma~\ref{lem:finer_partition_II} ensures that ${\rm{MN}}(\A_{\sim}, {\mathscr{G}}) \subseteq {\rm{MN}}(\A_{\sim}, {\mathscr{F}})$. But ${\rm{MN}}(\A_{\sim}, {\mathscr{F}}) =\, \sim$ on the canonical semi-automaton. Hence it suffices to prove that $\sim \, \, \subseteq {\rm{MN}}(\A_{\sim}, {\mathscr{G}})$.
By Lemma~\ref{lem:vertices_suffice}, it suffices to prove that $\sim_{V}\, \subseteq \underset{\!\!\!\!0}{\overset{\mathscr{G}}{\sim}_{V}}$. But by definition of ${\mathscr{G}}$, $x \sim_{V} y$ is equivalent to $x \underset{\!\!\!\!0}{\overset{\mathscr{G}}{\sim}_{V}}\, y$. 
\end{proof}

Whenever a vertex is reachable, the previous result yields a rather appealing (and presumably, familiar) description of an automatic relation in terms of one single subset of vertices.

\begin{corollary} \label{cor:accessible_automatic_description}
Let $G$ be a labelled directed graph with at least one reachable vertex $v$. A relation $\sim$ on $G$ is automatic if and only if $\sim\, = {\rm{MN}}(\A_{\sim}, \{ [v]_{\sim_{V}} \})$. 
\end{corollary}

\begin{proof}
By assumption, $S = \{ v \}$ is a minimal complete final system. Let ${\mathscr{F}}$ be the partition induced by the single class $[v]_{\sim_{V}}$. Theorem~\ref{th:automatic_as_mn} applies.
\end{proof}

This corollary shows that at least in the case of a reachable vertex, an automatic relation can be regarded as the graph-theoretic version of a Myhill-Nerode relation with respect to one distinguished subset of vertices (final states).

\subsection{A partial order on automatic relations} \label{subsec:partial_order_on auto_rel}

There is a natural partial order on automatic relations. We define it as follows: $(\sim_v, \sim_E) \leq (\sim'_v, \sim'_e)$ whenever $\sim_v \subseteq \sim'_v$ and $\sim_E \subseteq \sim'_e$.

\begin{prop}\label{pr:sim_order_to_morphism}
	Suppose that $\sim \leq \sim'$, then there is a directed emulator $\phi_{\sim,\sim'}: \sursim{G} \to G/{\sim'}$ with $\phi_{\sim,\sim'} \circ [-]_{\sim'} = [-]_{\sim}$. 
\end{prop}

\begin{proof}If $\sim \leq \sim'$, for all vertex $[x]_\sim \subseteq [x]_{\sim'}$ so that $[x]_\sim \mapsto [x]_{\sim'}$ is a function. The same holds for edges. Adjacency follows from the definitions. The equality $\phi_{\sim,\sim'} \circ [-]_{\sim'} = [-]_{\sim}$ is immediate.
\end{proof}

\begin{lemma} \label{lem:automatic_rels_upper_bound} Automatic relations have least upper bounds. 
\end{lemma}

\begin{proof}
Given  $(\sim_V, \sim_E)$ and $(\sim'_V, \sim'_E)$, let $\sim_V''$ and $\sim_E''$ be respectively the transitive closure of $\sim_V \cup \sim_V'$ and $\sim_{E}\cup \sim'_{E}$. 
Let us check that $(\sim_V'',\sim_E'')$ is an automatic relation. First, since $\sim_V$ and $\sim_V'$ are equivalence relations, the transitive closure of their union is an equivalence relation. The same argument holds for edges. 

We argue by induction on the length of the transitive closure that the relation $\sim''_E$ is compatible with $\sim''_V$. The base case is immediate. For the inductive step, suppose that an edge $e \sim_E e' \sim_E'' e''$ (the other case $e \sim'_E e' \sim_E'' e''$ is similar). Then $s_{G}(e) \sim_{V} s_{G}(e')$. By induction, $s_{G}(e') \sim_{V}'' s_{G}(e'')$. By transitivity, $s_{G}(e) \sim_{V}'' s_{G}(e'')$. Similarly $t_{G}(e) \sim_{V}'' t_{G}(e')$. 

For bisimilarity, suppose that a vertex $x \sim''_V s_G(f)$ for some edge $f$. The base case is again immediate. For the inductive step, we suppose that we have $x \sim_V x' \sim_V'' x'' = s_G(f)$ (the other case $x\sim_V' x' \sim''_V x''$ being symmetric). By induction, there is an edge $f'\sim''_E f$ with $s_G(f') = x'$. By bisimilarity, for $\sim$, there is an edge $f''$ such that $s_{G}(f'') = x$ and $f'' \sim_{E} f'$. By transitivity $f'' \sim''_E f$. 

By a (tedious but not difficult) induction on the transitive closure, one verifies that $(\sim''_V,\sim''_E)$ is actually the least upper bound. 
\end{proof}

The lower bound  $(\sim_V'',\sim_E'')$ of two relations  $(\sim_V, \sim_E)$ and $(\sim'_v, \sim'_E)$ is somewhat heavier to define. Set $f \sim_{E}'' f'$ if $f \sim_{E} f'$ and $f \sim_{E}' f'$. The relation on vertices is defined by
$x \sim''_{V} y$ if the following properties are satisfied:
\begin{enumerate}
\item[(i)] $ x \sim_{V} y \wedge x \sim_{V}' y$;
\item[(ii)] if $\arete{x}{e}{x'}$ and $y\sim_V x$ then there is an edge $\arete{y}{f}{y'}$ such that $e \sim_{E}'' f$.
\end{enumerate}

\begin{remark}Both conditions are necessary. If (ii) does not hold, the definition yields an equivalence relation that fails to be automatic (it fails to satisfy bisimilarity). For instance, consider the graph \begin{tikzpicture}[transform shape,scale=0.6,baseline=-1mm]
\node[etat] (A) at (0,0) {$v_0$};
\node[etat] (B) at (1.5,0) {$u_0$};
\node[etat] (C) at (3,0) {$w_0$};
\node[etat] (D) at (5,0) {$v_1$};
\node[etat] (E) at (6.5,0) {$u_1$};
\node[etat] (F) at (8,0) {$w_1$};
\draw[->] (B) -- (A);
\draw[->] (B) -- (C);
\draw[->] (E) -- (D);
\draw[->] (E) -- (F);
\end{tikzpicture}. Consider $\sim$ the vertex-induced automatic containing the three pairs $(u_0,u_1)$, $(v_0,v_1)$, $(w_0,w_1)$ and an other vertex-induced automatic relation build on $(u_0,u_1)$,$(v_0,w_1)$,$(w_0,v_1)$. We have $u_0 \sim u_1 \wedge u_0 \sim' u_1$ but this is not bisimilar with respect to $e\sim'' e' = e \sim e' \wedge e\sim' e'$ that is empty. 
\end{remark}

\begin{lemma}\label{le:aut_lower_bound}Automatic relations have greatest lower bounds.
\end{lemma}
\begin{proof}
We prove similarly that the lower bound $(\sim_V'',\sim_E'')$ is an automatic relation. It follows from the definition that it is an equivalence relation. The compatibility of $\sim_{E}''$ with respect to $\sim_{V}''$ immediately follows from the compatibility of $\sim_{E}$ (resp. $\sim_{E}'$) with respect to $\sim_{V}$ (resp. $\sim_{V}'$).Bisimilarity is a direct consequence of (ii).

Let us verify it is the greatest upper bound. Take $\sim''' \leq \sim$ and $\sim'''\leq \sim'$. By definition, $e\sim''' e' \Rightarrow e \sim e' \wedge e\sim' e'$. Thus $e\sim'' e'$. For the same reason, $x \sim''' x' \Rightarrow x \sim x' \wedge x\sim' x'$ so that (i) holds. Suppose that (ii) does not hold. Then, there is an edge $f\in \outE(x)$ such that for any edge $g\in \outE(y)$ either $f\not\sim g$ or $f\not\sim' g$. In both cases, $f \not\sim''' g$. But that means that $\sim'''$ has not the bisimilarity property for vertex $x$ with respect to $f$.
\end{proof}

\begin{prop} \label{prop:automatic_rels_lattice}
Automatic relations on a graph $G$ form a lattice. 
\end{prop}

Since the lattice is finite, there is a maximum element. The minimum is $(1_{V_G}, 1_{E_G})$. 

Looking at the lattice as a category, Proposition~\ref{pr:sim_order_to_morphism} shows that $\sim \mapsto [-]_\sim : G \to \sursim{G}$ is a functor from the lattice to  the coslice category, $\hom_{\Emu}(G, -)$. Actually, we have:

\begin{theorem}\label{th:terminal_in_slice}
In the coslice category, $\hom_{\Emu}(G, -)$, the directed emulator $[-]_\sim : G \to \sursim{G}$  with $\sim$ the maximum element of the lattice is a terminal object. 
\end{theorem}

\begin{proof}Suppose that $\phi : G \to H$ is a directed emulator. Then, there is an automatic relation $\sim'$ such that $\phi = \iota \circ [-]_{\sim'}$. Since $\sim$ is a maximum, by Proposition~\ref{pr:sim_order_to_morphism}, $[-]_\sim = \phi_{\sim',\sim} \circ [-]_{\sim'} = \phi_{\sim',\sim} \circ \iota^{-1} \circ \phi$. Uniqueness is a direct consequence of Proposition~\ref{pr:emu_right_cancel}.
\end{proof}

%
%

\section{From graphs to regular languages}

\subsection{Graphs as Semi-automata} \label{sub:graph-to-semi-auto}

\begin{definition} \label{def:tautological_transition_system}
Let $G$ be a directed graph. The canonical semi-automaton $A_{G} = (G, E_G, 1_{E_G})$ associated to $G$ with respect to the identity automatic relation ($\sim_{V} = 1_V$ and $\sim_{E} = 1_{E}$) is called 
 the {\emph{tautological semi-automaton}}.
\end{definition}

The morphism $(f,g) : G \to H$ between two directed graphs is sent to $\Gaut{(f,g)} = (f,g,g) : \Gaut{G} \to \Gaut{H}$. Indeed, we have $g \circ 1_{E_G} = 1_{E_H} \circ g$, thus is a morphism. So, the tautological assignement forms a functor $${\rm{A}}_{(\mathunderscore)}: \DG \to {\mathbf{Semi}}.$$

\begin{prop}\label{A_is_full_and_faithful}The functor $G \to \Gaut{G}$ is full and faithful.
\end{prop}

\begin{proof}Consider two morphisms $\phi = (f,g):G \to H$ and $\psi = (f',g') : G \to H$ such that $\Gaut{\phi} = (f,g,g) = (f',g',g')= \Gaut{\psi}$. Then, $f=f'$ and $g=g'$.  
	
Suppose now $(f, g, \alpha) : \Gaut{G} \to \Gaut{H}$.  It satisfies $\alpha \circ {\rm{id}}_{E} = {\rm{id}}_{E'} \circ g$, that is, $\alpha = g$ so that $(f, g, \alpha) = \Gaut{(f,g)}$. 
\end{proof}

\begin{lemma} \label{lem:tautological_properties}
The following properties hold:
\begin{enumerate}
\item[$(i)$] The identity $1_{\DG} : G \to G$ is a natural transformation $1_{\DG} \to \Autg{\Gaut{(-)}}$. 
\item[$(ii)$] There is a natural transformation $\epsilon : \Gaut{{\rm G}_{(-)}} \to 1_{\mathbf{Semi}}$.
\end{enumerate}
\end{lemma}

\begin{proof}
$(i)$ follows from the definition. 
For (ii), we define $\epsilon_\A = (1_{V_\A}, 1_{E_\A}, \ell_\A)$. It is a (relabelling) morphism $\Gaut{{\rm G}_\A} \to \A$. Indeed, due to (i), the two automata share the same graph: $\Autg{\Gaut{{\rm G}_\A}} = \Autg{\A}$. Thus $(1_{V_\A}, 1_{E_\A})$ is a graph morphism. 
Second, we have  $\ell_{\A_{{\Autg{\A}}}} = 1_{E_{\Autg{\A}} }  = 1_{E_\A} : E_{\A} \to E_{\A}$, again due to $(i)$. Thus, $\ell_{\A} \circ \ell_{\A_{{\rm G}_{\A}}} = \ell_\A \circ  1_{E_\A}$ that leads to Equation~\ref{eq:rel_labels_morphisms}. 

Let us verify naturality. Given $(f, g, \alpha) : \A \to \B$,  by definition, $\Gaut{\Autg{(-)}}(f,g, \alpha) = (f,g, g)$. We have to verify that the diagram commute:
$$ \xymatrix{
	{\A_{\Autg{\A}}} \ar[d]_{(f,g,g)} \ar[rr]^-{(1_{V_\A},1_{E_A},\ell_\A)} && \A  \ar[d]^-{(f, g, \alpha)} \\
	{\A_{\Autg{\B}}} \ar[rr]^-{(1_{V_\B},1_{E_B},\ell_\B)} && \B
}$$

For vertices and edges, this is trivial. The last equation is read $\alpha \circ \ell_\A = \ell_\B \circ g$ which is  Equation~\ref{eq:rel_labels_morphisms}. 
\end{proof}

\begin{corollary} \label{prop:tautological_functor_is_left_adjoint}
The tautological assignment $G \mapsto {\A}_{G}$ is left-adjoint to the forgetful functor ${\rm{G}}_{(\mathunderscore)}$: ${\mathbf{Semi}} \to \DG$.
\end{corollary}

A semi-automaton $\A$ is {\emph{complete}} (resp. {\emph{deterministic}})
 if given any state $q \in V_{\A}$, the map $\ell_q : {\rm OutE}(q) \to \Sigma_\A, \ e \mapsto \ell(e)$ is surjective (resp. injective\footnote{We rule out multiple transitions with the same source, target and label for a deterministic semi-automaton. (Not only this is consistent with the traditional definition, but this is required for
the next theorem to hold.) }). The following observation is a direct consequence of the definitions.

\begin{lemma} \label{lem:tautological_semi_is_deterministic}
For any directed graph $G$, the tautological semi-automaton ${\rm{A}}_{G}$ is deterministic.
\end{lemma}

\begin{lemma}\label{lem:complete_and_deterministic_outE}
	If $\A$ is complete and deterministic, then, for all $q \in V_\A$, $\outE(q) \simeq \Sigma_\A$.
\end{lemma}

\begin{proof}
	By definition, $\ell_q$ is both surjective and injective. 
\end{proof}

\begin{corollary}\label{cor:emu_is_cover_complete_deterministic}
	Suppose that $\phi :\A \to \B$ is strict morphism, $\A$ is both complete and deterministic and $G_\phi$ is a directed emulator, then, it is a cover.
\end{corollary}

\begin{proof} Let $\phi = (f,g,\alpha)$. For all $q \in V_\A$, we have $\ell_\B \circ g_q = \ell_q$ that is an isomorphism. Thus $g_q$ is an isomorphism. 
\end{proof}

\begin{lemma} \label{lem:direct_equiv_conformal}
	Suppose that $\phi: \A \to \B$ is a semi-automaton epimorphism. Assume the following conditions:
	\begin{enumerate}
		\item[$(1)$] The source semi-automaton $\A$ is complete;
		\item[$(2)$] The target semi-automaton $\B$ is deterministic.
	\end{enumerate}
	Then the morphism $G_\phi$ is a directed emulator.
\end{lemma}

\begin{proof}
	Suppose that $\phi = (f, g, \alpha)$ with $(f,g)$ an epimorphism. Then, $f$ is surjective. Let us check the outgoing edge lifting property.  Let $\arete{q_0}{e}{q'_0} \in E_{\B}$ be a transition in $\B$ and $p_0 \in f^{-1}(q_0)$. We look for a preimage of $e$ starting at $p_0$. Since $(f,g)$ is an epimorphism, there exists $p_1 \in f^{-1}(q_0)$ and $\arete{p_1}{\tilde{e}}{q'_1} \in E_{\A}$ such that $g(\tilde{e}) = e$, $f(p_1) = q_0$ and $f(q'_1) = q'_0$. Furthermore, $\alpha(\ell_{A}(\tilde{e})) = \ell_{\B}(e)$. If $p_0 = p_1$, we are done. Otherwise, since $\A$ is complete, there is some edge $\arete{p_0}{\tilde{e}'}{p_1} \in E_{\A}$ starting from $p_0$ with the same label $\ell_{\A}(\tilde{e})$. Consider its image $g(\tilde{e}')$. We have $s_\B(g(\tilde{e}')) = q_0$ and $\ell_{\B}(g(\tilde{e}')) = \alpha(\ell_{\A}(\tilde{e}')) = \alpha(\ell_{\A}(\tilde{e})) = \ell_{\B}(e)$. Therefore the edge $g(\tilde{e}')$ has same source and same label as the edge $g(\tilde{e})$. Since $\B$ is deterministic, this implies that $g(\tilde{e}') = g(\tilde{e}) = e$ and we are done.
\end{proof}

\subsection{Finite state automata} 

\label{sec:finie_state_auto}

An {\emph{automaton}} $\A$ is a semi-automaton endowed with a set of states $I \subseteq V_G$ that is called the set of \emph{initial states} and a set $F \subseteq V_{\A}$ of \emph{final states}. A state $q\in V_G$ is \emph{accessible} if $q$ is reachable from an initial state, that is, if there is a  state $q_0 \in I$ and a walk starting at $q_0$ and ending at $q$. A state $q \in V_G$ is \emph{co-accessible} if $q$ is co-reachable to a final state, that is, if there is a state $q_{1} \in F$ and walk starting at $q$ and ending at $q_1$. \emph{We suppose in this paragraph that any state $q \in V_G$ is accessible.}


The standard definition considers automata with a unique initial state. We call them \emph{finite state automata}. To denote the full class, we speak about multi-input state automata to stress the fact that multiple inputs are allowed. $\Maut$ denotes the full class. 

\begin{definition}\label{def:automaton_morphism}
Let $\A$ and $\B$ be two automata. A {\emph{morphism between automata}} $\A \to \B$ is a
semi-automaton morphism $(f, g, \alpha)$ verifying 
\begin{enumerate}
	\item $I_\A = f^{-1}(I_\B)$ and \label{eq:init_state}
	\item $F_\A = f^{-1}(F_\B)$. \label{eq:final_state}
\end{enumerate}
In other words, $q \in I_\A \Leftrightarrow f(q) \in I_\B$ and $q\in F_\A \Leftrightarrow f(q) \in F_\B$.
\end{definition}
The identity is an automaton morphism. Automata morphism compose componentwise. We denote \textbf{Auto} the category of automata and their morphisms. 

\begin{lemma}
	The assignment that forgets the sets of initial and final states induces a forgetful faithful functor ${\mathbf{Auto}} \to {\mathbf{Semi}}$.
\end{lemma}

\begin{remark}
	The forgetful functor is not full. For instance, the induced map $H_1 = {\rm Hom}_{\mathbf{Auto}}\left( \begin{tikzpicture}[->,transform shape, scale=0.7, initial text={}, baseline=-1mm]
	\node (I0) at (0,0.75) {};
	\node (F0) at (1.5,0.75) {};
	\node[state, inner sep=3pt, minimum size=12pt ](G0) at (0,0) {i};
	\node[state,  inner sep=2pt, minimum size=12pt](G1) at (1.5,0) {f};
	\path [] (G0) edge [bend left] node [above=0.5mm] {$e:a$} (G1) ;
	\path [] (G1) edge [bend left] node [below=0.5mm]
	{$e':a$} (G0);
	\draw (I0) -- (G0);
	\draw (G1) -- (F0);
	\end{tikzpicture},\  \begin{tikzpicture}[->,transform shape, scale=0.7, initial text={}, baseline=-1mm]
	\node (I0) at (0,0.75) {};
	\node (F0) at (1.5,0.75) {};
	\node[state, inner sep=3pt, minimum size=12pt ](G0) at (0,0) {i};
	\node[state,  inner sep=2pt, minimum size=12pt](G1) at (1.5,0) {f};
	\path [] (G0) edge [bend left] node [above=0.5mm] {$e:a$} (G1) ;
	\path [] (G1) edge [bend left] node [below=0.5mm]
	{$e':a$} (G0);
	\draw (I0) -- (G0);
	\draw (G1) -- (F0);
	\end{tikzpicture}\right) \to
	{\rm{Hom}}_{\mathbf{Semi}}\left( \begin{tikzpicture}[->,transform shape, scale=0.7, initial text={}, baseline=-1mm]
	\node[state, inner sep=2pt, minimum size=12pt ](G0) at (0,0) {i};
	\node[state,  inner sep=2pt, minimum size=12pt](G1) at (1.5,0) {f};
	\path [] (G0) edge [bend left] node [above=0.5mm] {$e:a$} (G1) ;
	\path [] (G1) edge [bend left] node [below=0.5mm]
	{$e':a$} (G0);
	\end{tikzpicture},\  \begin{tikzpicture}[->,transform shape, scale=0.7, initial text={}, baseline=-1mm]
	\node[state, inner sep=2pt, minimum size=12pt ](G0) at (0,0) {i};
	\node[state,  inner sep=2pt, minimum size=12pt](G1) at (1.5,0) {f};
	\path [] (G0) edge [bend left] node [above=0.5mm] {$e:a$} (G1) ;
	\path [] (G1) edge [bend left] node [below=0.5mm]
	{$e':a$} (G0);
	\end{tikzpicture}\right)
	= H_2$ is not onto. Indeed, $H_1$ contains only the identity while $H_2$ contains also  the semi-automaton morphism induced by exchanging the two vertices.
\end{remark}

A ({\emph{length $n$}}) {\emph{computation}} in $\A$ starting at state $q$, ending at state $q'$, is a walk $\pi = e_0 e_1 \ldots e_n$ in $G_\A$. The {\emph{label}} of the computation $\pi$ is $\ell(\pi) = \ell(e_0)\ell(e_1)\cdots \ell(e_n) \in \Sigma_\A^{*}$. A computation in an automaton is {\emph{successful}} if it starts at some initial state and ends at some final state. 

\begin{definition} \label{def:language}Given an automaton $\A$ the subset $L(\A) \subseteq \Sigma_\A^*$ of the words $w$ for which there is a a successful computation $\pi$ such that $\ell(\pi) = w$ is called the language represented by the automaton.
\end{definition}

From Kleene's Theorem, we know that the languages defined by (finite state and multi-input state) automata are the regular languages. 

\begin{lemma}\label{lem:image_word_via_morphism}
	The image of a successful computation $\pi$ in $\A$ by a semi-automaton morphism $(f,g, \alpha) :\A \to \B$ is a succesful computation $\pi' = g(\pi)$ in $\B$. Furthermore, the label of the image of the computation is the image of the label of the computation:
	$\ell_\B(\pi') = \ell_\B(g(\pi)) = \alpha(\ell_A(\pi))$.
\end{lemma}

\begin{proof}The function $g$ preserving adjacency, $g(\pi)$ is a walk within $G_\B$.   By Lemma~\ref{lem:image_word_via_morphism} and items (1) and (2) of Definition~\ref{def:automaton_morphism}, initial and final states are preserved. The equality of images is a direct consequence of Equation~\ref{eq:rel_labels_morphisms}.
\end{proof}

We record a consequence:

\begin{lemma} \label{lem:automaton-morphism-and-computation}
	If there is an automaton morphism $(f,g,\alpha):\A \to \B$ then
	$\alpha(L(\A)) \subseteq L(\B)$. In particular, if there is a strict morphism
	$\A \to \B$ then $L(\A) \subseteq L(\B)$.
\end{lemma}

\begin{prop} \label{prop:strict-conformal-implies-same-language}
	Suppose that $\phi:\A \to \B$ is a morphism between two automata $\A$ and $G_\phi$ is a directed emulator, then $\alpha(L(\A)) = L(\B)$ where $\phi = (f,g,\alpha)$. In particular, if $\phi$ is a strict morphism, $L(\A) = L(\B)$.
\end{prop}

\begin{proof} Any successful walk in $\B$ is the image of a walk in $\A$. The preimage in $\A$ of an initial (resp. final) state in $\B$ being itself initial (resp. final), the walk in $\A$ is successful.  Thus, $\alpha(L(\A)) \supseteq L(\B)$. With the preceding Lemma, we conclude $\alpha(L(\A)) = L(\B)$. 
\end{proof}

\begin{remark}
	In Sakarovitch's book, a strict automaton morphism that induces a directed emulator
	morphism is called a totally surjective morphism \cite[Chap. II, Def.~3.2]{Sakarovitch}.
\end{remark}

\begin{theorem} \label{theo:functor_properties}
	Suppose that $\phi: \A \to \B$ is an automaton morphism. Assume the following conditions:
	\begin{enumerate}
		\item[$(1)$] The source  automaton $\A$ is complete;
		\item[$(2)$] The target automaton $\B$ is deterministic.
	\end{enumerate}
	Then the two assertions are equivalent:
	\begin{enumerate}[(i)]
		\item the morphism $G_\phi$ is an epimorphism,
		\item  $G_\phi$ is a directed emulator morphism.
	\end{enumerate}
\end{theorem}

\begin{proof}
	
	Lemma~\ref{lem:direct_equiv_conformal} shows that $(i) \Rightarrow (ii)$.  By Lemma~\ref{lem:direct-emulator-is-epimorphism}, $(ii) \Rightarrow (i)$.
\end{proof}

\begin{corollary}\label{cor:epi_complete_deterministic_is_cover}
	If $\phi: \A \to \B$ is a {strict} epimorphism between complete and deterministic
	automata, then $G_\phi : G_\A \to G_\B$ is a directed covering map.
\end{corollary}

\begin{proof} First, by the preceding theorem, it is a directed emulator morphism. Then, by Corollary~\ref{cor:emu_is_cover_complete_deterministic}, it is a covering map. 
\end{proof}

We recall that, given a deterministic finite state automaton $\A$, there is a minimal complete deterministic finite state automaton, denoted $\A_{\rm{min}}$, such that $L(\Amin) = L(\A)$ together with the canonical projection $\A \to \Amin$ that sends equivalent states to their equivalence class; furthermore, $\Amin$ is unique up to automaton (strict) isomorphism. 

\begin{corollary} \label{cor:epimorphism-to-covering}
	Let $\A$ be a deterministic automaton and let $\pi: \A \to \Amin$ be the canonical
	epimorphism to the minimal automaton. There exists an automaton $\tilde{\A}$ such that the following properties are satisfied:
	\begin{enumerate}
		\item[$(1)$] $\tilde{\A}$ is complete and deterministic;
		\item[$(2)$] $\tilde{\A}$ contains $\A$ as a subautomaton;
		\item[$(3)$] $L(\tilde{\A}) = L(\A)$;
		\item[$(4)$] The underlying graph morphism  ${\mathrm{G}}_{\tilde{\A}} \to
		{\mathrm{G}}_{\Amin}$ is a covering morphism.
	\end{enumerate}
\end{corollary}

\begin{proof} Each time there is a state $q$ in $\A$ and a letter $a \in\Sigma$ without edges $\arete{q}{e}{q'} \in E_\A$ with $\ell_\A( e ) = a$, take $p = \pi( q )$. Since $\Amin$ is complete, there is an outgoing edge $\arete{p}{e'}{p'}$ with $\ell_{\Amin}( e')  = a$.  Take some node $q'$ in $Q_\A$ in the fibre over $p'$ (it exists since the morphism is onto). We add to $\A$ a new edge $e''$ between $p$ and $p'$ with label $a$ and we set $\pi(e'') = e'$. After modification, the automaton $\A$ still verify the hypotheses of the Corollary. We continue the process until $A$ is complete, call the result $\tilde{\A}$.
	Since $\tilde{\A}$ is complete and deterministic,  Corollary \ref{cor:epi_complete_deterministic_is_cover} leads to (4).
\end{proof}

We now seek a reconstruction of an automaton morphism from a directed emulator morphism.
The following lemma is the first step.

\begin{lemma}
	Consider a directed graph ${\rm{G}}_\A$ induced by an automaton $\A$. Suppose that there is a directed emulator
	morphism $\varphi: G' \to \rm{G}_\A$.
	Then there exists an automaton $\A'$  and a strict automaton morphism $\pi: \A' \to \A$  such that
	\begin{enumerate}
		\item[$(1)$] $G_\pi = \varphi$; 
		\item[$(2)$] $L(\A') = L(\A)$.
	\end{enumerate}
\end{lemma}

Here, we do not ask the automaton $\A$ to be deterministic nor complete.

\begin{proof} We begin with the commutative diagram:
	$$\xymatrix{
		\A_{G'} \ar^{\A_\varphi}[r] \ar^{G}[d] & \A_{G_\A} \ar_{G}[d] \ar^{\epsilon_\A}[r] & \A \ar^{G}[d] \\
		G' \ar[r]^{\varphi} & G_\A \ar^{1}[r] & G_\A
	}$$
	
	By Proposition~\ref{conformal_decomposition}, the morphism $\epsilon_\A \circ A_\varphi = \pi \circ \lambda$ for a strict morphism $\pi$ and $\lambda$ a relabelling. Then, set $\A' = \lambda ( \A_{G'} )$. We have $G_{\A'} = G_{A_{G'}} = G'$ and $G_\pi = \varphi$. However, as defined, $\A'$ is a semi-automaton. We define $I_{\A'} = \pi^{-1}( I_\B )$ and $F_{\A'} = \pi^{-1}( F_\A)$ so that $\pi$ is actually an automaton morphism. We conclude with Lemma~\ref{prop:strict-conformal-implies-same-language}. 
\end{proof}

\begin{lemma} \label{lem:from-emulator-to-automaton}
	Consider a directed graph $G = {\mathrm{G}}_\A$ induced by some deterministic automaton $\A$. Suppose that there is a directed emulator
	morphism $G' \to G$. Then there exists a deterministic automaton $\A''$ and a strict automaton epimorphism $\pi: \A'' \to \A$ such that
	\begin{enumerate}
		\item[$(1)$] $G_\pi: {\mathrm{G}}_{\A''} \to {\mathrm{G}}_\A$ is a directed covering morphism;
		\item[$(2)$] ${\mathrm{G}}_{\A''}$ is a subgraph of $G'$;
		\item[$(3)$] $g(\A'') \leq g(G')$.
	\end{enumerate}
\end{lemma}

\begin{proof}
	By the previous lemma, there is an automaton $\B$ together with a
	strict epimorphism $\B \to \A$ inducing the directed emulator
	morphism $G' \to G$. By Lemma \ref{lem:directed-emulator-contains-directed-cover},
	one can extract from the directed emulator $G'$ a directed cover $G'' \subseteq G'$ over $G$.
	This determines a subautomaton $\A''$ of $\B$ with underlying directed graph ${\mathrm{G}}_{\A''} = G''$ that is a subgraph of $G'$.
	Therefore $g(\A'') = g(G'') \leq g(G')$. It remains to verify that $\A''$ is deterministic. Since the strict epimorphism $\A'' \to \A$ induces a covering morphism $G'' \to G$, for any state $q' \in Q_{\A'}$ and its image $q \in Q_{\A}$, the
	induced map ${\rm{OutE}}(q') \to {\rm{OutE}}(q)$ is a bijection between sets
	of labelled outgoing transitions. Since $\A$ is deterministic, $\ell_{\A}|_{{\rm{OutE}}(q)}$ is injective and so must be $\ell_{\A''}|_{{\rm{OutE}}(q')}$. Therefore $\A''$ is deterministic.
\end{proof}

We can get rid of the multiple edges in the previous lemma, i.e., a simple graph theoretical version of the previous lemma holds.

\begin{lemma} \label{lem:from-directed-simple-emulator-to-automaton}
	Consider a directed graph $G = {\mathrm{G}}_\A$ induced by some deterministic automaton $\A$. Suppose that there is a directed emulator morphism $\phi: H \to R(G)$. Then there exists a deterministic automaton $\A'$ and a strict automaton epimorphism $\A' \to \A$ such that
	\begin{enumerate}[(1)]
		\item The induced
		directed emulator morphism ${\mathrm{G}}_\A' \to {\mathrm{G}}_\A$ is a
		directed covering morphism;
		\item $g(\A') \leq g(H)$.
	\end{enumerate}
\end{lemma}

\begin{proof}The two morphisms $G \stackrel{\rho_G}{\to} R(G) \stackrel{\phi}{\leftarrow} H$ are directed emulators. Via Lemma~\ref{lem:pull_backs}, we have the diagram with ${\pi_1}_{\mid}$ a directed emulator:
	$$\xymatrix{
		G \times_{R(G)} H \ar[r]^-{\pi_{2}|} \ar[d]_-{\pi_{1}|} & H \ar[d]^-{\phi} \\
		G \ar[r]^-{\rho_G} & R(G)
	}$$
Observe that $R(G \times_{R(G)} H) \simeq R(G) \times_{R(R(G))} R(H) \simeq R(G) \times_{R(G)} R(H) \simeq R(H)$, the first equality being due to Lemma~\ref{le:R_preserve_pullback}. Thus $g(G \times_{R(G)} H) =  g(H)$ via Lemma~\ref{lem:R_preserves_genus}. We conclude  by Lemma~\ref{lem:from-emulator-to-automaton}.  
\end{proof}

Finally we get rid of loops in the following sense. 

\begin{lemma}  \label{lem:from-directed-cover-over-simple-selfloopless-to-automaton}
	Consider a directed graph $G = {\mathrm{G}}_\A$ induced by some deterministic automaton $\A$. Suppose that there is a simple directed cover
	morphism $H \to {\rm{Exc}}(R(G))$. Then there exists a deterministic automaton $\A'$ and a strict automaton epimorphism $\A' \to \A$ such that
	\begin{enumerate}[(1)]
		\item The induced
		directed emulator morphism ${\mathrm{G}}_\A' \to {\mathrm{G}}_\A$ is a
		directed covering morphism;
		\item $g(\A') \leq g(H)$.
	\end{enumerate}
\end{lemma}

\begin{proof}
	Apply Lemma \ref{lem:from-cover-over-simple-excized-to-cover-over-original} to obtain a directed cover $H'$ over $R(G)$ such that $g(H') = g(H)$. Applying Lemma \ref{lem:from-directed-simple-emulator-to-automaton} yields the result.
\end{proof}

\section{The genus of a regular language}\label{sec:genus_regular_language}

In this section, we state and prove two major results of the paper (Theorem~\ref{th:genus-and-emulator} and Corollary~\ref{cor:genus_equivalent_emulator}). We also discuss the relationship with the undirected setting.

\subsection{Main results}

Let us recall the definition that we introduced in \cite{BD16}.

\begin{definition}
The {\emph{genus}} $g(L)$ of a regular language $L$ over alphabet $A$ is the minimum of all genera of finite state deterministic automata computing $L$.
\end{definition}

Before we proceed to the main result, we mention three important observations in the definition.

\begin{remark}[Determinism]
The word ``deterministic'' is essential in the definition of the genus, for it is known that any regular language has a genus $0$ (planar) nondeterministic automaton that computes it, a nice result due to R.~.V.~Book and A.~K.~Chandra~\cite{BC76}.
\end{remark}

\begin{remark}[Single-input versus multi-input]
We would like to emphasize that the definition of genus here is confined to regular languages computed by automata with one single initial state (single-input automata), in accordance to our convention in this paper (cf. \ref{sec:finie_state_auto}). Taking into account multi-input (deterministic) automata leads to a richer and more complex notion of genus, studied in \cite{BD23}.
\end{remark}

\begin{remark}[Completeness]
It can be proved (\cite{BD16}) that there always exists a {\emph{complete}} deterministic
automaton $\A$ such that $g(L) = g(\A)$. Basically, given a deterministic (but not complete) automaton, for any missing transition, one may add a transition to a fresh trash state whose outgoing transitions are loops. That leaves the genus unchanged.
\end{remark}

%
%
%

The directed graph $G(L)$ {\emph{associated to a regular language}} $L$ is the directed graph underlying the minimal automaton $A_{\rm{min}}(L)$ canonically associated to $L$: $G(L) = {\mathrm{G}}_{\A_{\rm{min}}(L)}$.
For any regular language $L$, $g(L) \leq g(G(L))$.

\begin{theorem} \label{th:genus-and-emulator}
Let $L$ be a regular language. Let $n \in {\mathbb{N}}$. The following assertions are equivalent:
\begin{enumerate}
\item[$(1)$] $g(L) \leq n$;
\item[$(2)$] The directed graph $G(L)$ has a directed cover $G$ of genus $g(G) \leq n$;
\item[$(3)$] The directed graph $G(L)$ has a directed emulator $G$ with $g(G)\leq n$;
\item[$(4)$] The directed simple graph $R(G(L))$ has a directed simple cover $G$ such that $g(G) \leq n$;
\item[$(5)$] The directed simple graph ${\rm{Exc}}(RG(L))$  has a directed cover $G$ such that
$g(G) \leq n$.
\end{enumerate}
\end{theorem}

\begin{proof}
Suppose that $L$ has genus $g(L) \leq n$. There is some finite deterministic complete automaton $\A$ computing $L$ such that $g(\A) \leq n$. This automaton comes naturally with an automaton epimorphism $\A \to \A_{\rm{min}}(L)$. Applying the functor ${\mathrm{G}}_{(-)}$ yields (Theorem \ref{theo:functor_properties}) a directed covering epimorphism ${\mathrm{G}}_\A \to G(L)$. Hence $(1) \Longrightarrow (2)$.
The implication $(2) \Longrightarrow (3)$ is obvious.
The functor $R$ preserves directed emulation (Lemma \ref{lem:functor_R_preserves_directed_emulators}) and genus (Lemma \ref{lem:R_preserves_genus}), thus $(3) \Longrightarrow (4)$. Excision also preserves the genus, thus $(4) \Longrightarrow (5)$. Suppose (5). Applying Lemma~\ref{lem:from-directed-cover-over-simple-selfloopless-to-automaton} provides a deterministic automaton $\A$ such that $L(\A) = L$ and
$g(\A) \leq g(G) \leq n$. 
Therefore $g(L) \leq g(\A) \leq n$. This proves $(1)$.
\end{proof}

\begin{corollary}
If $L, L'$ are two regular languages such that
${\rm{Exc}}(RG(L)) = {\rm{Exc}}(RG(L'))$, then $g(L) = g(L')$.
\end{corollary}

\begin{corollary} \label{cor:subgraph_and_languages}
Let $L, L'$ be regular languages such that ${\rm{Exc}}(RG(L'))$ is a subgraph of ${\rm{Exc}}(RG(L))$. Then $g(L') \leq g(L)$.
\end{corollary}

\begin{proof}
A directed emulator $H$ of ${\rm{Exc}}(RG(L))$ of minimal genus contains a directed emulator $H'$ of ${\rm{Exc}}(RG(L'))$ as a subgraph. Hence $g(L) \geq g(H') \geq g(L')$.
\end{proof}

Corollary~$\ref{cor:subgraph_and_languages}$ can be used to bound genera of languages as well.

\begin{corollary} \label{cor:size_less_or_equal_to_six_is_planar}
Any regular language $L$ of size $|L|_{\rm{set}} \leq 6$ is planar.
\end{corollary}

\begin{proof}
Consider the regular language $Z_{6}$ that consists of words $w = a_1 a_2 \cdots a_n$ on the alphabet ${\mathbb{Z}}/6{\mathbb{Z}}$ such that the sum $\sum_{i} a_{i}$ of its letters is $0$ mod $6$. 
The underlying graph of its minimal deterministic automaton is the complete simple directed graph of size 6:
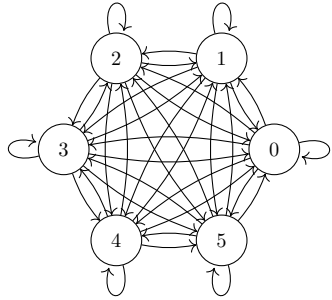
\begin{center}
 \begin{tikzpicture}[->,transform shape, scale=0.7, initial text={}, baseline=-1mm]]
\node[state](G0) at (2,0) {$0$};
\node[state] (G1) at (1,1.73) {$1$};
\node[state] (G2) at (-1,1.73) {$2$};
\node[state] (G3) at (-2,0) {$3$};
\node[state] (G4) at (-1,-1.73) {$4$};
\node[state] (G5) at (1,-1.73) {$5$};
\path [] 
(G0) edge[bend right=10] node [right=1mm] {} (G1)
(G1) edge[bend right=10] node [above=1mm] {} (G2)
(G2) edge[bend right=10] node [left=1mm] {} (G3)
(G3) edge[bend right=10] node [left=1mm] {} (G4)
(G4) edge[bend right=10] node [below=1mm] {} (G5)
(G5) edge[bend right=10] node [right=1mm] {} (G0)
(G0) edge[bend right=10] node [below=1mm] {} (G2)
(G1) edge[bend right=10] node [below=1mm] {} (G3)
(G2) edge[bend right=10] node [right=1mm] {} (G4)
(G3) edge[bend right=10] node [above=1mm] {} (G5)
(G4) edge[bend right=10] node [above=1mm] {} (G0)
(G5) edge[bend right=10] node [left=1mm] {} (G1)
(G0) edge[bend right=10] node {} (G3)
(G1) edge[bend right=10] node {} (G4)
(G2) edge[bend right=10] node {} (G5)
(G3) edge[bend right=10] node {} (G0)
(G4) edge[bend right=10] node {} (G1)
(G5) edge[bend right=10] node {} (G2)
(G0) edge[bend right=10] node [right=1mm] {} (G5)
(G1) edge[bend right=10] node [above=1mm] {} (G0)
(G2) edge[bend right=10] node [left=1mm] {} (G1)
(G3) edge[bend right=10] node [left=1mm] {} (G2)
(G4) edge[bend right=10] node [below=1mm] {} (G3)
(G5) edge[bend right=10] node [right=1mm] {} (G4)
(G0) edge[bend right=10] node [right=1mm] {} (G4)
(G1) edge[bend right=10] node [above=1mm] {} (G5)
(G2) edge[bend right=10] node [left=1mm] {} (G0)
(G3) edge[bend right=10] node [left=1mm] {} (G1)
(G4) edge[bend right=10] node [below=1mm] {} (G2)
(G5) edge[bend right=10] node [right=1mm] {} (G3)
(G0) edge[loop right] (G0)
(G1) edge[loop above] (G1)
(G2) edge[loop above] (G2)
(G3) edge[loop left] (G3)
(G4) edge[loop below] (G4)
(G5) edge[loop below] (G5);
\end{tikzpicture}
\end{center}
This graph has a planar directed emulator (for this and other facts about $Z_6$, see \cite[\S 2.1]{BD19}). Let $L$ be any regular language whose minimal automaton $\A$ has size at most $6$. The simple graph $R(G_\A) = R(G(L))$ is then a subgraph of $R(G(Z_6))$, thus (Corollary~\ref{cor:subgraph_and_languages}) has a planar directed emulator. Hence by Theorem~\ref{th:genus-and-emulator}, $L$ is planar.

\end{proof}

\begin{corollary}
Let $L_1$ and $L_2$ be two regular languages on disjoint alphabets. Then $g(L_1 \cup L_2) \geq \max(g(L_1), g(L_2)).$
\end{corollary}

\begin{proof}
The minimal automaton $\A_{\rm{min}}(L_1 \cup L_2)$ for $L_1 \cup L_2$ contains both the minimal automaton $\A_{\rm{min}}(L_1)$ and the minimal automaton $\A_{\rm{min}}(L_2)$ as subgraphs.
\end{proof}


The {\emph{Language Genus Problem}} is the following: given a regular language $L$ and $n \in \N$, the answer is YES if $g(L) \leq n$, otherwise NO.

The {\emph{Directed Emulation Genus Problem}} is: given a directed graph $G$ and $n \in \N$,  YES if there is a directed emulator $G'$ of $G$ such that $g(G') \leq n$, otherwise NO.

Theorem~\ref{th:genus-and-emulator} allows to reduce the problem of the determining the genus of a language to a graph-theoretic problem in terms of directed emulators.

\begin{corollary}\label{cor:genus_equivalent_emulator}
The {\emph{Language Genus Problem}} has a solution if and only if the {\emph{Directed Emulation Genus Problem}} restricted to co-reachable directed graphs has a solution.
\end{corollary}

\begin{proof}
Theorem~\ref{th:genus-and-emulator} shows a solution for the Directed Emulation Genus Problem implies a solution for the Language Genus Problem. We need to show the converse. Suppose that $G=(V,E)$ is a strongly connected directed graph. Consider its associated tautological semi-automaton $\A_{G}$ (see Definition~\ref{def:tautological_transition_system}). By assumption, there is a vertex $v_0 \in V$ such that all other vertices are reachable from it. We define $v_0$ to be the initial state and all vertices $v \in V$ to be final states. Let $\A$ be the resulting automaton. By Lemma~\ref{lem:tautological_semi_is_deterministic}, $\A$ is deterministic. Since $v_0$ is co-reachable, that every state is accessible (from $v_0$) in $\A$. Since every state is final, every state is trivially co-accessible (to itself). Furthermore, since every outgoing edge has a unique label, there cannot be two distinct and equivalent states, so $\A$ is minimal. We have $G = G(L(\A))$. By Theorem~\ref{th:genus-and-emulator}, $L(\A)$ has genus $g$ if and only if $G$ has a directed emulator of genus $g$.
\end{proof}

The \emph{planarity problem} is the genus problem (applied to directed graphs or regular languages) restricted to $n = 0$. 
	
\begin{corollary} \label{cor:planarity_pb} The planarity problems for co-reachable directed graphs and for regular languages respectively are equivalent. 
\end{corollary}

\begin{remark}
Corollary~\ref{cor:planarity_pb} was first proved by D.~Kuperberg.
\end{remark}

\subsection{Directed vs undirected}

We state in this paragraph other applications of Theorem~\ref{th:genus-and-emulator} in relation to our study of the differences between directed and undirected emulators.

Our first observation is that Theorem~\ref{th:genus-and-emulator} yields a bound for genus by means of {\emph{undirected}} emulators.

\begin{corollary} \label{cor:genus_bound_with_undirected_emulator}
Let $L$ be a regular language and let $U_{L} = U({\rm{Exc}}(RG(L)))$ be the undirected graph obtained from $G(L)$ by simplification and excision. Let $H \to U_{L}$ be any undirected emulator of $U_{L}$. Then $g(L) \leq g(H)$. In particular, if $U_{L}$ has a planar emulator then $L$ is planar.
\end{corollary}

\begin{proof}
Suppose there is an emulator $\pi : H \to  U_L$. Then, by Lemma~\ref{lem:from_emulator_to_directed_emulator}, there is a directed emulator $\vec{\pi} : \vec{H} \to \rm{Exc}( \rm{ R }( G(L) )$ (note that there is a direction such that $\vec{U}_{L} = \rm{Exc}( \rm{ R }( G(L) )$). Therefore by Theorem~\ref{th:genus-and-emulator} $g(L) \leq g(\vec{H}) = g(H)$.
\end{proof}

For the next corollary, for a directed graph $G$, denote ${\rm{Em}}_{\DG}(G)$ the set of all finite directed emulators of $G$. For an undirected graph $H$, denote ${\rm{Em}}_{\DSG}(H)$ the set of all finite emulators of $H$.


\begin{corollary} \label{cor:bounds_directed_undirected}
Let $G$ be a simple loopless directed graph. Then
$$ \underset{\widetilde{G} \in {\rm{Em}}_{\DG}(G)}{\min}\ g(\widetilde{G}) \leq \underset{H \in {\rm{Em}}_{\DSG}(U(G))}{\min}\ g(H).$$
\end{corollary}

\begin{remark}
There are languages (resp. directed graphs) such that the inequality in Corollary~\ref{cor:genus_bound_with_undirected_emulator} (resp. Corollary~\ref{cor:bounds_directed_undirected}) is strict. In other words, the upper bound for the genus of directed emulators given by undirected emulators may be not reached. In particular, the converse of the last statement in Corollary~\ref{cor:genus_bound_with_undirected_emulator} does not hold. Consider the  language $L$ defined on the alphabet ${\mathbb{Z}}/7{\mathbb{Z}}$ as the set of three-letter words $abc$ $(a,b,c \in {\mathbb{Z}}/7{\mathbb{Z}})$ such that $a+b+c = 0$ mod $7$. The underlying graph $G(L)$ of the minimal deterministic automaton $\A_{L}$ for $L$ is depicted below.
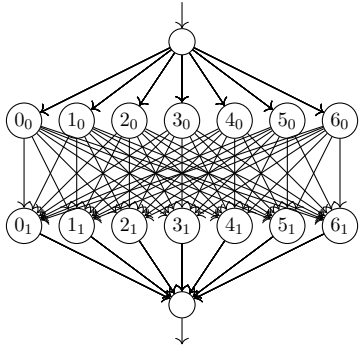
\begin{center}
 \begin{tikzpicture}[->,transform shape, scale=0.7, initial text={}, baseline=-1mm]
  \tikzstyle{every state}=[inner sep=2pt,minimum size=2pt]
\node[initial above, state,  inner sep=5pt, minimum size=5pt](G0) at (3,0.5) {};
\node[state] (H0) at (0,-1) {$0_{0}$};
\node[state] (H1) at (1,-1) {$1_{0}$};
\node[state] (H2) at (2,-1) {$2_{0}$};
\node[state] (H3) at (3,-1) {$3_{0}$};
\node[state] (H4) at (4,-1) {$4_{0}$};
\node[state] (H5) at (5,-1) {$5_{0}$};
\node[state] (H6) at (6,-1) {$6_{0}$};
\node[state] (I0) at (0,-3) {$0_1$};
\node[state] (I1) at (1,-3) {$1_1$};
\node[state] (I2) at (2,-3) {$2_1$};
\node[state] (I3) at (3,-3) {$3_1$};
\node[state] (I4) at (4,-3) {$4_1$};
\node[state] (I5) at (5,-3) {$5_1$};
\node[state] (I6) at (6,-3) {$6_1$};
\node[accepting below, state,  inner sep=5pt, minimum size=5pt] (K0) at (3,-4.5) {};
\foreach \x in {0,...,6}
   \foreach \y in {0,...,6}
{ \path[] (G0) edge[] node[]{} (H\x);
	\path[] (H\x) edge[] node[]{} (I\y);
  \path[] (I\x) edge[] node[]{} (K0);
}
\end{tikzpicture}
\end{center}

\begin{corollary}\label{cor:diff_directed_undirected}
The directed graph $G(L)$ has a planar emulator and the undirected graph $U(G(L))$ has no planar emulator.
\end{corollary}

\begin{proof}
As the language $L$ is finite, there is a finite deterministic (planar) tree that recognizes it, so $L$ is planar: $g(L) = 0$. By Theorem~\ref{th:genus-and-emulator}, the first assertion follows. For the second assertion, we need the following observations.

\begin{prop} \label{lem:easy_not_planar}
Let $L$ be an $m$-letter regular language such that $m \geq 3$. Suppose that the outdegree at each vertex of $G(L)$ is $m$ and that $U(G(L))$ has no cycles of length $\leq 2$. Then $g(L) \geq 1$.
\end{prop}

\begin{remark}
This is a slight improvement of \cite[Corollary 3.1]{BD19}. 
\end{remark}

\begin{proof}
This is a consequence of the genus formula \cite[Theorem 5]{BD16} recalled below. Consider a minimal embedding of a directed emulator $H$ of $G(L)$ in a closed surface. Since there is no cycle of length $1$ and $2$, each face of the embedding has length at least $3$ (the proof is the same as that of the claim in \cite[Proof of Prop. 1]{BD16}). Let $f_{i}$ be the number of $i$-faces of the embedding. According to the genus formula,
\begin{align*}
g(L)  = g(H) & = 1 - \frac{m+1}{m}f_1 - \frac{1}{2m}f_2 + \frac{m-3}{4m}f_3 + \frac{2m-4}{4m} f_4 + \cdots  \\
 & = 1 + \underbrace{\frac{m-3}{4m}}_{\geq 0}f_{3} + \underbrace{\frac{2m-4}{4m}}_{\geq 0} f_{4} + \cdots
\\
& \geq 1.
\end{align*}
\end{proof}

\begin{lemma} \label{lem:nonplanar_language}
Let $Z_{7}^{1, 2, 3}$ be the language on the alphabet $\{1,2,3\} \subset \mathbb{Z}/7{\mathbb{Z}}$ that consists of words $a_1 a_2 \cdots a_n$ such that $\sum_{i} a_{i} = 0$ mod $7$. Then $g(Z_{7}^{1, 2, 3}) \geq 1$.
\end{lemma}

\begin{proof}
The minimal complete deterministic automaton $\A$ for $Z_{7}^{1, 2, 3}$ is seen to be defined by $Q = {\mathbb{Z}}/7{\mathbb{Z}}$, $0$ being the initial and final state, and transitions $\arete{i}{j}{i+j}$ for $i \in {\mathbb{Z}}/7{\mathbb{Z}}$ and $j \in \{1, 2, 3 \} \subset {\mathbb{Z}}/7{\mathbb{Z}}$. The underlying graph $G(L)$ is easily seen to have outdegree $m$ at each vertex and to have no cycles of length $< 3$. Therefore Lemma \ref{lem:easy_not_planar} applies.
\end{proof}

For the next lemma, we need to define the \emph{directed cycle contraction} (see \cite{KiZh14}). Let $G = (V,E)$ be a directed graph and $c = e_1 e_2 \cdots e_n$ a directed cycle  in $G$ (the map $\{1, \ldots, n \} \to E, i \mapsto e_i$ is injective and $s(e_1) = t(e_n)$). Let $E_{c}$ denote the subset of edges $e$ in $E$ incident to $c$, that is, such that $s(e) = s(e_i)$ or $t(e) = t(e_i)$ for some $1 \leq i \leq n$. Let $w$ be a new vertex. Define a bijective map by sending each edge $e \in E_{c} - \{ e_1, \ldots, e_n \}$ to a new edge $e'$ where the vertex in $c$ incident to $e$ is replaced by $w$. Denote the image $E_{w}$. The contraction of $G$ along $c$ is the new directed graph $G_{c} = (V',E')$ where $V' = (V - \{ s(e_1), \ldots, s(e_n) \}) \cup \{ w \}$ and $E' = (E - E_{c}) \cup E_{w}$.

\begin{prop} \label{lem:cycle_contraction_does_not_increase_genus}
If $G$ has a directed emulator of genus $\leq g$, then $G_{c}$ has a directed emulator of genus $\leq g$.
\end{prop}

\begin{proof} \cite[Theorem~12]{Denis}.
\end{proof}

We now return to the proof of Corollary~\ref{cor:diff_directed_undirected}. By  Lemma~\ref{lem:double_respect_emulators}, $H \to U(G(L))$ is an (undirected) emulator morphism if and only if the induced map $\double{H} \to \double{U(G(L))}$ is a directed emulator morphism. In the directed graph $\double{U(G(L))}$, for each $i \in {\mathbb{Z}}/7{\mathbb{Z}}$, contract the directed cycle 
$c_{i} = \begin{tikzpicture}[->,transform shape, scale=0.7, initial text={}, baseline=-1mm]
\tikzstyle{every state}=[inner sep=2pt,minimum size=2pt] 
\node[state] (M) at (0,0) {$i_{0}$};
\node[state] (N) at (2,0) {$i_{1}$}; 
\path (M) edge[bend left] (N);
\path	  (N) edge[bend left] (M);
\end{tikzpicture}$. (Note that all these cycles are disjoint in $\double{U(G(L))}$.) Let $G = G_{c_1, \ldots, c_7}$ be the resulting directed graph and let $\tilde{G}$ be any directed emulator of $G$. Note that $G$ contains $\A_{Z_{7}^{1,2,3}}$. Therefore
\begin{align*}
g(\double{U(G(L))}) & \geq g(\tilde{G}) &  \hbox{(Lemma~\ref{lem:cycle_contraction_does_not_increase_genus})}\\
& \geq g(Z_{7}^{1,2,3}) &  {\hbox{(Corollary \ref{cor:subgraph_and_languages})}}\\
& \geq 1 & \hbox{(Lemma~\ref{lem:nonplanar_language})}.
\end{align*}
\end{proof}
\end{remark}

We now consider the undirected version of the Emulation Genus Problem. The {\emph{Emulation Genus Problem}} is: given a connected undirected graph $G$ and $n \in \N$, answer YES if there is an emulator $G'$ of $G$ such that $g(G') \leq n$, otherwise NO. 

\begin{corollary}\label{cor:sol_for_directed_implies_sol_for_undirected}
The {\emph{Emulation Genus Problem}} has a solution if the {\emph{Directed Emulation Genus Problem}} has a solution.
\end{corollary}

\begin{proof}
Let $G$ be a graph. Assume that the Directed Emulation Genus Problem has a solution. It suffices, then, to prove that there is an emulator $G'$ of $G$ such that $g(G') \leq n$ if and only if there is a directed emulator $H$ of $\double{G}$ such that $g(H) \leq n$. 

Suppose that $\phi: G'\to G$ is an emulator such that $g(G') \leq n$. By Lemma~\ref{lem:double_respect_emulators}, $\double{\phi}: \double{G'} \to \double{G}$ is a directed emulator. Furthermore, it is clear that $g(\double{G'}) = g(G')$. 
	
Conversely, suppose that $\phi: G'\to \double{G}$ is a directed emulator with $g(G') \leq n$. According to Lemma~\ref{lem:double_almost_adjoint_to_U}, there is an emulator $G'' \to G$ with $g(G'') = g(G') \leq n$. 
\end{proof}

\begin{remark}
Corollary~\ref{cor:sol_for_directed_implies_sol_for_undirected} may be a little surprising at first. Indeed, the existence of a directed cover of genus $g$ is equivalent to the existence of a directed emulator of genus $g$ (Proposition~\ref{prop:directed_is_cover}), a fact that turns out to be wrong in the undirected setting (after the works of P.~Hlin{\v{e}}n{\'y} \cite{Hl99}, Y.~Rieck and Y.~Yamashita \cite{RY09}, respectively) as discussed above (\S \ref{subsec:undirected_emulators}). However, Corollary~\ref{cor:sol_for_directed_implies_sol_for_undirected} (more precisely the key lemma used in the proof above) does not contradict this. Typically, the solution provided by Corollary~\ref{cor:sol_for_directed_implies_sol_for_undirected}  is an emulator and not a cover, even if one started from a directed cover $G' \to \double{G}$. Indeed, the key lemma (Lemma~\ref{lem:double_almost_adjoint_to_U}) cannot be used in general to build an undirected cover from a directed cover. See our remarks there (Remark \ref{rem:does_not_work_for_covers}). 
\end{remark}

\begin{remark} \label{rem:nonconstructive_sol_for_undirected}
On the one hand, the existence of an emulator of an undirected graph of genus $\leq n$ is preserved under edge contraction, see \cite[\S 2]{FL88}. It follows that the Emulation Genus Problem is decidable (albeit not constructively). On the other hand, the existence of a directed emulator of genus $\leq n$ of a directed graph is not even preserved under edge contraction, see \cite[\S 2]{Denis}.
\end{remark}

\section{Conclusion}

We have shown that the Language Genus Problem is decidable if and only if the Directed Emulation Genus Problem is decidable. However, we do not have yet a complete proof of decidability. On another direction, we proved 
that the {\emph{Emulation Genus Problem}} has a solution if the {\emph{Directed Emulation Genus Problem}} has a solution. The former problem is known to have a theoretical solution (Remark~\ref{rem:nonconstructive_sol_for_undirected}). 
A general approach, suggested by a natural generalization of Corollary~\ref{cor:subgraph_and_languages}, consists in properly defining directed minors and proving a ``directed graph minor'' theorem analogous to the celebrated graph minor theorem of Robertson and Seymour \cite[\S 10.5]{RS04}. This is the approach aimed at in \cite{Denis}. Suitable operations are defined at the level of the underlying graph $G(L)$ of the minimal deterministic automaton $\A_{\rm{min}}(L)$ that are non-increasing on the genus of the language $L$ (hence, by Theorem~\ref{th:genus-and-emulator}, of any directed cover of $G(L)$). One should note, however, that the operations (and hence the ``directed graph minors'') are less elementary than the operations for undirected graphs. Even if this approach would be successful, one would furthermore need to find the minors of nonplanar emulable directed graphs. We would face the kind of issues that are discussed by M.~Chimani, M.~Derka, P.~Hlin{\v{e}}n{\'y} and M.~Klus{\'a}{\v{c}}ek in their  article~\cite{CDHK13} (on undirected graphs). 
 In any case, the complete relationship between undirected and directed  emulators, beyond Corollary~\ref{cor:sol_for_directed_implies_sol_for_undirected},   is intriguing and seems quite a challenging question.

Finally, we note that the fields of linear logic and automata theory are closely related. See for instance  the early paper by R.~Statman \cite{Statman74} who introduced the genus of a proof by means of a directed graph $G(\Pi)$ associated to a proof $\Pi$ (see also \cite{Carbone09}). We expect that the questions raised in this paper have applications or analogs in linear logic \cite{Girard87}. Furthermore, developments in linear logic include the assignment of more general objects than graphs to proof nets, empowering the full arsenal of categorical topological invariants to (suitable categories of) linear logic (see e.g., \cite{mellies07}).

\bibliography{main5}
\bibliographystyle{plain}

\end{document}